\documentclass[12pt, draftclsnofoot, onecolumn ]{IEEEtran}
\usepackage[dvips]{color}
\usepackage{epsf}
\usepackage{times}
\usepackage{epsfig}
\usepackage{graphicx}

\usepackage{amsmath}
\usepackage{amssymb}
\usepackage{amsxtra}
\usepackage{amsthm}
\usepackage{bbm}

\usepackage{here}
\usepackage{rawfonts}
\usepackage{times}
\usepackage{url}
\usepackage{cite}
\usepackage{comment}
\usepackage[utf8]{inputenc}
\usepackage{caption}
\usepackage{subcaption}

\usepackage{pstricks}
\usepackage{algorithm}
\usepackage{algpseudocode}
\usepackage{lipsum}
\usepackage{mathtools}

\newtheorem{corollary}{Corollary}
\newtheorem{theorem}{\bf Theorem}

\newtheorem{proposition}{\bf Proposition}

\makeatletter 
 
\begin{document} 
\title{ Distributed Conditional Generative Adversarial Networks  (GANs) for Data-Driven Millimeter Wave Communications in UAV Networks 
\thanks{A preliminary version of this work appears in the proceedings of IEEE ICC 2021 \cite{zhang2021distributed}.}} 
\author{ 
	\IEEEauthorblockN{Qianqian Zhang$^1$,  
		Aidin Ferdowsi$^1$,  
		Walid Saad$^1$,
		and Mehdi Bennis$^2$
	}
	
	\IEEEauthorblockA{\small 
		$^1$Bradley Department of Electrical and Computer Engineering, Virginia Tech, VA, USA,	
		Emails: \url{{qqz93,aidin,walids}@vt.edu}.  \\
		$^2$Centre for Wireless Communications, University of Oulu, Finland, Email: \url{mehdi.bennis@oulu.fi}. 
	}
} 
\maketitle

\setlength{\columnsep}{0.55cm}

\begin{abstract}
In this paper, a novel framework is proposed to perform  data-driven  air-to-ground channel estimation for millimeter wave (mmWave) communications in an unmanned aerial vehicle (UAV) wireless network. 
First, an effective channel estimation approach is developed to collect  mmWave channel information,   allowing each UAV to train a stand-alone channel model via a conditional generative adversarial network (CGAN) along each beamforming direction.   
Next, in order to expand the application scenarios of the trained channel model into a broader spatial-temporal domain,   a cooperative framework, based on a distributed CGAN architecture, is developed, allowing each UAV to collaboratively learn the mmWave channel distribution  in a fully-distributed  manner.   
To guarantee an efficient learning process,  necessary and sufficient conditions for the optimal UAV network topology that maximizes the learning rate for cooperative channel modeling are derived, and the optimal CGAN learning solution per UAV is subsequently characterized, based on the distributed network structure.    
Simulation results show that the proposed distributed CGAN approach is robust to the local training error at each UAV.  
Meanwhile, a larger  airborne network size requires more communication resources per UAV to guarantee an efficient learning rate.  
The results also show that, compared with a stand-alone CGAN  without information sharing and two other distributed schemes, namely: A multi-discriminator CGAN and a federated-learning CGAN method, the proposed distributed CGAN approach yields a higher modeling accuracy while learning the environment, and it achieves a larger average data rate  in the online performance of UAV downlink mmWave communications. 
 
\end{abstract}


\IEEEpeerreviewmaketitle

\section{Introduction}

Millimeter wave (mmWave) frequency bands are a  pillar of next-generation wireless systems as they will enable ultra-high-speed communications and airborne wireless networks \cite{akdeniz2014millimeter}. 
In order to overcome the fast attenuation of mmWave signals, multiple-input multiple-output (MIMO) technologies with highly-directional beamforming  are employed so as to increase the cell throughput and improve the communication reliability.  
Compared with the sub-6 GHz spectrum, the higher frequency of mmWave yields a shorter coherence time for the wireless channels.  Therefore, mmWave communication links are more time-sensitive and require frequent channel measurements \cite{zhang2020millimeter}. 
Meanwhile, highly-directional beamforming requires accurate knowledge of angle-of-arrivals (AoAs) and angle-of-departures (AoDs) for the propagation paths in order to achieve beam alignment between transmitter and receiver \cite{noh2020fast}.  
However, both the channel estimation and beam training  can incur  heavy communication overhead and reduce the spectrum  efficiency  \cite{dong2019deep}. 
Therefore, in order to improve the transmission performance, it is essential to have an accurate model to characterize a mmWave  link and estimate its underlying MIMO channels.

Compared with a terrestrial communication network, mmWave channel modeling for an airborne, drone-based wireless  system is more challenging \cite{dabiri2020analytical}. 
Airborne wireless networks have been increasingly considered as a suitable platform to deploy mmWave communications, due to a high possibility of line-of-sight (LOS) link states.    
For example, an unmanned aerial vehicle (UAV)-based station (BS) can dynamically adjust its location so as to maintain a LOS channel with transceivers \cite{zhang2019reflections}.     
However,  a mobile UAV BS must often provide connectivity to a much larger geographical area compared with a typical  terrestrial cellular BS. 
Moreover, compared to a terrestrial channel, the air-to-ground (A2G) channel  includes more model parameters, such as the 3D location and dynamic orientation of the UAV, which makes the channel modeling process  more challenging.  
Meanwhile, given that the mmWave channel response is time and location dependent, the statistical model generated from one communication environment experienced by a UAV, at a given time and spatial coordinate, cannot be flexibly generalized into  other temporal or spatial settings \cite{yang2019generative}.
As a result, traditional channel modeling methods, such as ray-tracing, can become very difficult and time-consuming to measure mmWave A2G channels efficiently.  
Indeed, most current A2G channel models are calibrated at sub-6 GHz frequencies, while a standard-defined A2G model over mmWave bands, as well as the experimental data,  has been very limited \cite{xia2020generative}.    
However, for a UAV BS, it is essential to  have an accurate A2G channel model to estimate the mmWave link state, and, thus, save pilot training time and transmit power for efficient communications.    
In order to address the challenges of mmWave A2G channel modeling, a \emph{data-driven} approach can be applied, where a UAV BS collects the A2G channel information during its cellular service, and then, build a  stochastic channel model to estimate the long-term channel parameters.   
Such a UAV can also  collaborate with neighboring UAVs to build a generalized spatial-temporal map of the mmWave environment.

\subsection{Related Works}
In order to capture the stochastic characteristics of mmWave channels, a number of data-driven modeling approaches were developed in \cite{akdeniz2014millimeter,dong2019deep}, and \cite{han2016two,alkhateeb2019deepmimo,xia2020generative,khawaja2019survey,polese2020experimental,cheng2020modeling,khawaja2017uav}.  
Traditional methods, such as spatial-temporal correlation \cite{akdeniz2014millimeter} and compressed sensing \cite{han2016two}, were investigated for characterizing mmWave MIMO transmissions. 
The authors in  \cite{dong2019deep} developed a neural network approach for characterizing mmWave transmissions and estimating MIMO channels. 
A deep learning dataset  is developed in \cite{alkhateeb2019deepmimo}  to extract the propagation feature of mmWave-based communication links. 
However, all of the proposed modeling frameworks  in \cite{dong2019deep,akdeniz2014millimeter,han2016two}, and \cite{alkhateeb2019deepmimo} focus on a terrestrial mmWave transmission scenario, and their results are not applicable for A2G channel modeling, due to distinct propagation environments.   
To specifically address the characteristics of airborne mmWave communications, some recent works in  \cite{xia2020generative} and \cite{polese2020experimental,cheng2020modeling,khawaja2019survey,khawaja2017uav} investigate the UAV-related mmWave channel from different perspectives. 
The authors in \cite{khawaja2019survey} provided a comprehensive survey on A2G propagation channel modeling for both the microwave and mmWave spectrum bands. 
The authors in \cite{xia2020generative} developed a generative neural network to predict the mmWave link state and model statistical channel parameters between a UAV and a ground BS.  
An empirical propagation loss model is proposed in \cite{polese2020experimental} based on an extensive  measurement for UAV-to-UAV communications at 60 GHz, and  a traditional ray tracing  method is applied in  \cite{cheng2020modeling}   to build a geometry-based stochastic model for  UAV-to-vehicle communications at 28 GHz. 
Furthermore, the received signal strength and delay spread of mmWave  transmissions is analyzed in \cite{khawaja2017uav} to  provide  further details for A2G channels.  
However,  all of the prior art in \cite{xia2020generative} and \cite{khawaja2019survey,polese2020experimental,cheng2020modeling,khawaja2017uav} studies the characteristics of mmWave channel models based on a single and local dataset, and, thus, the generated channel model is constrained by a limited amount of channel samples and  a few dedicated measurement environments. As such, these existing models cannot be used as a general and standardized model for A2G mmWave channels.

In order to extend the channel model to large-scale application scenarios, a promising solution is to use a cooperative modeling approach  with multiple, distributed channel datasets.  
In a recent work \cite{elbir2020federated}, the authors developed  a federated learning (FL) framework  to train the channel model from  distributed data sources.  However, the centralized network topology of the FL framework requires a global controller for information aggregation,  and, thus, it cannot operate in a fully distributed network as is the case in an airborne network.  
The work in \cite{park2020communication} characterized a time-varying channel model via continuous data exchange in a distributed wireless system. 
However, sharing the raw channel data in a real-time manner yields heavy communication overhead. 
Furthermore, beyond the discriminative models in \cite{dong2019deep} and \cite{elbir2020federated}, 
a generative machine learning model is applied in recent works \cite{yang2019generative} and \cite{ye2020deep} to model the wireless channel. 
The authors in \cite{yang2019generative} proposed a generative adversarial network (GAN) framework to model the wireless channel based on massive raw data, and the work in \cite{ye2020deep} employed a conditional GAN  to represent unknown channels to enable the encoding and modulation optimization, given the pilot training information.  
However, all of the prior works in \cite{yang2019generative} and \cite{elbir2020federated,park2020communication,ye2020deep} do not focus on the characteristic of mmWave frequencies or A2G wireless links. Thus, their results are not applicable to the mmWave channel estimation in the UAV communications. 
Given that  a generative model can learn the application range of the channel model from the temporal-spatial information in the dataset while training the channel features, it provides a better learning framework compared with a discriminative approach (e.g., such as those in \cite{dong2019deep} and \cite{elbir2020federated}).  
Remarkably, there are no fully distributed generative learning frameworks developed to deal with the problem of data-driven mmWave channel modeling in prior works. As such, to fill this gap, in this work, a fully-distributed cooperative  generative learning model will be developed  to characterize the environment of A2G mmWave links.   
\subsection{Contributions}
The main contribution of this paper is a novel framework that can perform  data-driven channel modeling for mmWave communications in a distributed UAV network.  
In particular, a  learning approach, based on a distributed conditional generative adversarial network (CGAN) is proposed for the UAV network to jointly learn the mmWave A2G channel characteristics from multiple, distributed datasets. 
In summary, our key contributions are:
\begin{itemize}
	\item First, we develop an effective channel measurement approach to collect  real-time A2G channel information over mmWave frequencies, allowing each UAV to train a stand-alone channel model via a CGAN at each beamforming direction.  
	
	\item Next,   to expand the application scenarios of the trained channel model into a broader spatial-temporal domain, we propose a cooperative learning framework, based on the distributed framework of brainstorming GANs \cite{ferdowsi2020brainstorming}. 
	This distributed generative approach allows  each UAV to learn the channel distribution from other agents in a fully distributed manner, while  characterizing  an underlying  distribution of the mmWave channels based on the entire  channel dataset of all the UAVs.  
	In order to avoid revealing the real measured data or the trained channel model to other agents, each UAV shares synthetic channel samples that are generated from its local  channel model in each  iteration.  	The proposed approach does not require any control center, and it can accommodate different types of neural networks.
	
	\item  To guarantee an efficient learning process in the distributed UAV system, we analytically derive  the probability of learning completion for the distributed CGAN learning at each iteration. 
	Then, we theoretically derive the necessary and sufficient conditions for the optimal UAV-to-UAV communication topology that maximizes the learning rate for cooperative channel modeling.  
	Finally, based on the structure of the distributed UAV network, we characterize the optimal CGAN learning solution per UAV.    
\end{itemize}

Simulation results show that the proposed distributed CGAN approach is resistant and robust to the local training error of each UAV.   
When the airborne network size becomes larger, more communication resources per UAV  are required to guarantee an efficient learning rate. 
Meanwhile,  in each iteration, by sharing more generated samples, the learning rate of the distributed CGAN scheme will increase, but the data transmission duration will also be larger.  
To ensure an efficient data transmission, a better wireless link state or a larger transmit power will be  required so as to improve the transmission rate.  
The results also show that, compared with a local CGAN  without information sharing and other distributed schemes, such as the multi-discriminator CGAN and the FL-based CGAN methods, the proposed distributed CGAN approach yields a higher modeling accuracy in the learning result, and it achieves a higher average data rate  in the online performance of UAV downlink mmWave communications.

The rest of this paper is organized as follows. 
Section \ref{sysModel} presents the communication model and data collection.  
The CGAN learning framework, distributed UAV network, and problem formulation are presented in Section \ref{channelModeling}. The optimal network topology and learning solutions are derived in Section \ref{solution}. 
Simulation results are shown in Section \ref{simulation}. 
Conclusions are drawn in Section \ref{conclusion}.

\section{Communication Model and Data Collection }\label{sysModel}

\subsection{Millimeter Wave Channel Model}
Consider an airborne cellular network, in which a set $\mathcal{I}$ of UAVs provide mmWave downlink communications  to  ground user equipment (UE).  
We assume that each UAV is equipped with an uniform linear array  of $M$ antennas, and the steering vector of the UAV's transmit antennas  is given by  
$\boldsymbol{a}_t(\phi^t) = [1,  e^{j\frac{\pi}{\lambda}  \sin(\phi^t)}, \cdots, e^{j(M-1)\frac{ \pi}{\lambda} \sin(\phi^t)} ]^T$,  
where $\lambda$ is the carrier wavelength, and $\phi^t \in [0,2\pi)$ is the AoD. 
Meanwhile, each UE is equipped with  an uniform linear array  of  $N$ antennas,  and the receiver's antenna vector  is  $\boldsymbol{a}_r(\phi^r)$ $ = $ $[1, $ $e^{j\frac{ \pi}{\lambda} \sin(\phi^r)}, $ $ \cdots, $ $e^{j(N-1)\frac{ \pi}{\lambda} \sin(\phi^r)} ]^T $ with  $\phi_k^r$ being the AoA. 
Consequently, the MIMO channel matrix 
$\boldsymbol{H} \in \mathcal{C}^{N\times M}$  can be given by \cite{han2016two} 
\begin{equation}\label{mimoChannel}
\boldsymbol{H} = \sum_{l=1}^{L} \alpha_l \boldsymbol{a}_r(\phi_l^r) \boldsymbol{a}_t^H (\phi_l^t),
\end{equation}
where  $(\cdot)^H$ is the conjugate transpose,  $L$ is the number of distinct paths, and  $\alpha_l  \in \mathcal{C}$,  $\phi_l^t$,  and $ \phi_l^r$ are  the  complex channel gain,  AoD, and  AoA of path $l$, respectively.   
Given the fact that the A2G channel over mmWave is very sensitive to blockage and has few scattering links, the value of $L$ will be much smaller than $M\times N$. 
Meanwhile, since a massive MIMO mmWave array can provide a narrow beam that eliminates much of the multipath \cite{swindlehurst2014millimeter}, we can assume that both the UAV and each UE apply a perfect  directional radiation technology for beam training purposes,   
such that both the transmitter's and receiver's antennas have only one main lobe without any side lobes. Thus,  the mmWave channel consists of a single path, i.e. $L=1$, which is either line-of-sight (LoS), reflected none-line-of-sight (NLoS), or complete outage. 
This assumption is supported by experimental results in \cite{khawaja2017uav} and \cite{khawaja2018temporal}, where a single path-component is observed from the  $60$ GHz UAV communications in the suburban and rural environments. 
Meanwhile, the transceiver design in \cite{sadhu20177} and other prior works in \cite{song2015adaptive,khalili2020optimal,hussain2018energy} can further support the use of a single-path assumption for general mmWave systems.  

Now, we consider a UAV located at 3D coordinates  $\boldsymbol{x}$ and a UE located at 3D coordinates $\boldsymbol{y}$, where  $\boldsymbol{x}$ and  $\boldsymbol{y}$ are available via a global positioning system (GPS) module, equipped at each UAV and UE, respectively. Then, at the service time $t$, the A2G MIMO channel matrix in (\ref{mimoChannel}) can be rewritten as
\begin{equation}
\boldsymbol{H}(\boldsymbol{x},\boldsymbol{y},t,\phi^t,\phi^r) =   \alpha(\boldsymbol{x},\boldsymbol{y},t,\phi^t,\phi^r)  \boldsymbol{a}_r(\phi^r) \boldsymbol{a}_t^H (\phi^t),
\end{equation}
where the channel gain $\alpha$   is jointly determined by the AoD-AoA pair $ \boldsymbol{\phi}  \triangleq  (\phi^t,\phi^r) $ of the mmWave path,  as well as  the communication environment  $(\boldsymbol{x},\boldsymbol{y},t)$. 
Here,  the service time $t$ is based on a local clock at the UAV.  

\subsection{Channel Estimation and Data Collection}
In a directional transmission system,  pilot training is necessary to align the beam direction between the UAV and its served UE.
 Thus, 	the the A2G mmWave channel information can be measured during the beam training stage.   
Due to the narrow beam, it is necessary to exploit a pre-determined codebook for beam alignment between the UAV and the UE.  
Let $K$ be the length of the codebook, and $(\boldsymbol{w}_k,\boldsymbol{q}_k)$ be the $k$-th pair of beamforming and combining vectors in the codebook. Then, the received pilot signal at the UE for the $k$-th pilot training is 
\begin{equation}\label{receivedPilot} \vspace{-0.1cm}
	r_k = \sqrt{P} \boldsymbol{q}_k^H \boldsymbol{H}_k  \boldsymbol{w}_k + \boldsymbol{q}_k^H \boldsymbol{n},
\end{equation}
where  $P$ is the transmit power of pilot symbols at the UAV, and  $\boldsymbol{n} \sim \mathcal{CN}(\boldsymbol{0},\sigma^2_{\textrm{UE}}\boldsymbol{I}_N)$ is the noise vector. 
Here, we assume that each UAV and UE will have  a  digital beamforming phased array   such that the direction information  of signal departure $\phi_k^t$ and signal arrival $\phi_k^r$ during the $k$-th pilot training can be uniquely determined by the beamforming vector $\boldsymbol{w}_k$ and  the combining vector $\boldsymbol{q}_k$, respectively. 
Therefore, the AoD can be given as  a function of the beamforming vector  $\phi_k^t(\boldsymbol{w}_k)$, and the  AoA as $\phi_k^r(\boldsymbol{q}_k)$.
The relationship between the beamforming direction and the digital weights of the antenna arrays has been widely studied and mathematically derived in   \cite{ noh2020fast},  \cite{bas2017real}, and   \cite{sadhu20177}. 
Then, let $\otimes$ be the Kronecker product, and $\textrm{vec}(\cdot)$ be the  vectorization of a matrix. 
The received pilot signal in (\ref{receivedPilot}) can be rewritten as
\begin{equation}\label{ps} \vspace{-0.1cm}
\begin{aligned}
r_k &= \sqrt{P} (\boldsymbol{w}_k^T \otimes \boldsymbol{q}_k^H) \textrm{vec}(\boldsymbol{H}_k)  + \boldsymbol{q}_k^H \boldsymbol{n}, \\
&= \sqrt{P} (\boldsymbol{w}_k^T \otimes \boldsymbol{q}_k^H) [\boldsymbol{a}_t^{*}(\boldsymbol{w}_k)  \otimes \boldsymbol{a}_r(\boldsymbol{q}_k)] \alpha_k(\boldsymbol{x},\boldsymbol{y},t,\boldsymbol{\phi}_k)  + \boldsymbol{q}_k^H \boldsymbol{n} , \\
&= \beta_k \alpha_k(\boldsymbol{x},\boldsymbol{y},t,\boldsymbol{\phi}_k)  + \boldsymbol{q}_k^H \boldsymbol{n},  
\end{aligned}
\end{equation} 
where $(\cdot)^{T}$ is the transpose, $(\cdot)^{*}$ is complex conjugate,  and $\beta_k = \sqrt{P} (\boldsymbol{w}_k^T \otimes \boldsymbol{q}_k^H) (\boldsymbol{a}_{t,k}^{*}  \otimes \boldsymbol{a}_{r,k}) \in \mathcal{C}$. 
After receiving $\{r_k\}_{k=1,\cdots,K}$, each UE will send the pilot training information to the UAV via  a sub-$6$ GHz uplink \cite{feng2018spectrum}.   
Note that the beamforming vector $\boldsymbol{w}_k$ and the combining vector $\boldsymbol{q}_k$  are chosen from a pre-determined codebook which is known by the UAV for beam training purposes, and the antenna steering vectors  $\boldsymbol{a}_t$ and  $\boldsymbol{a}_r$ can be uniquely determined based on  $\boldsymbol{w}_k$ and  $\boldsymbol{q}_k$ directly. 
Thus, the value of $\{ \beta_k\}_{\forall k}$ is completely known by the UAV.      
Therefore, according to  pilot signals in (\ref{ps}), the channel gain from the UAV located at $\boldsymbol{x}$ towards a UE located $\boldsymbol{y}$ during time $t$ with an AoD-AoA  pair $\boldsymbol{\phi}_k$ can be estimated via 
\begin{equation}\vspace{-0.1cm}
\begin{aligned}
 \tilde{\alpha}_k(\boldsymbol{x},\boldsymbol{y},t,\boldsymbol{\phi}_k) =  r_k \beta_k^{-1}  =  {\alpha}_k(\boldsymbol{x},\boldsymbol{y},t,\boldsymbol{\phi}_k) + \tilde{n}_k,
\end{aligned}
\end{equation}  
where $\tilde{n}_k = \beta^{-1}_k \boldsymbol{q}_k^H \boldsymbol{n} $ is the uncorrelated estimation error.

During the airborne cellular service, the channel gain  $\tilde{\alpha}_k$ can be measured and collected by each UAV over a spatial-temporal domain, with $K$ different pairs of antenna steering directions.  
We denote the channel dataset of a given UAV $i$ as $\mathcal{S}_i = \{\boldsymbol{s}_n,\boldsymbol{\phi}_n \}_{n = 1,\cdots,S_i}=\{\boldsymbol{x}_n,\boldsymbol{y}_n,t_n,\tilde{\alpha}_n, \boldsymbol{\phi}_n \}_{n = 1,\cdots,S_i}$, 
where $\boldsymbol{s}_n =\{\boldsymbol{x}_n,\boldsymbol{y}_n,t_n,\tilde{\alpha}_n\} $ is a channel sample, $\boldsymbol{\phi_n}$ is the AoA-AoD information associated with $\boldsymbol{s}_n$, and $ S_i = |\mathcal{S}_i| $ is the total number of channel samples. 
Here,  note that the pilot training stage is inevitable for beamforming transmissions, and the resulting  pilot signals are usually discarded after channel measurement.  In this work, we allow each UAV to keep the pilot signals, collect channel information, and build its own dataset $\mathcal{S}_i$. 
Given that all the channel components are available from pilot signals, the proposed data collection process does not require any additional communication overhead to  form the dataset $\mathcal{S}_i$ at each UAV $i$. 

Based on $\mathcal{S}_i$, each UAV can build its own model for estimating  A2G mmWave links in its dedicated measurement area. 
However,  traditional channel modeling approaches, such as ray tracing and regression, cannot provide a  suitable modeling framework for A2G mmWave links for several reasons. 
First, different from terrestrial wireless channels, there are limited studies on the characteristics of  A2G links over the mmWave spectrum (e.g., see \cite{khawaja2017uav} and \cite{khawaja2018temporal}, however, those works focus on the small-scale temporal and spatial characteristics of mmWave channel, but do not develop tractable  models).  
As a result, it is difficult to find a rigorous data-driven channel model to optimize mmWave A2G channel parameters, using  well-known regression methods \cite{cheng2020modeling}.  
Meanwhile,  the amount of channel  samples that each UAV owns is usually not  sufficient to build an accurate stochastic model that properly  captures the amplitude, phase and directional features of the MIMO channel response over a large spatial-temporal domain. 
For example, the altitude of a UAV  largely determines the LOS probability towards UEs, and, hence, an accurate estimation of the mmWave link state will require a large amount of measurement data, which mandates  cooperative learning.  
Moreover, the mmWave channel characteristics are location and time dependent.  
In order to build a general and standardized  A2G mmWave model, the channel data must be collected from multiple distinct communication environments, 
and the corresponding channel parameters will span a larger set of possible values, which brings more challenges for the channel modeling accuracy.  
To address the aforementioned  challenges for mmWave A2G channel estimation, 
we will introduce a data-driven deep learning approach with cooperative information sharing, such that an accurate channel model can be developed by each UAV  over a large-scale spatial and temporal domain.

\section{Distributed Channel Modeling via Conditional GANs}\label{channelModeling}

\subsection{Conditional GAN for Channel Modeling}

Given the channel dataset $\mathcal{S}_i$, each UAV $i$ can train a local channel  model, based on deep neural networks (DNNs).
For example, a  discriminative learning model, such as in \cite{dong2019deep} and \cite{elbir2020federated},  can be trained to take a spatial-temporal pair as input and outputs a complex channel gain.  
This discriminative framework enables a UAV to predict the mmWave channel given any new  input, however, it fails to capture any additional information from $\mathcal{S}_i$, other than the channel gain value. 
Indeed, based on the dataset  $(\boldsymbol{x}_n,\boldsymbol{y}_n,t_n)_{\forall n} \subset \mathcal{S}_i$, we can acquire the geographic area that UAV $i$ has previously visited and the time interval during which UAV $i$  measures the channel information. 
The underlying distribution of the spatial-temporal pairs in $\mathcal{S}_i$ will define the application range of the trained channel model.  
In order to jointly capture the applicable spatial-temporal domain from  $\mathcal{S}_i$ while learning the channel model, 
we propose a generative approach for the mmWave A2G channel modeling. 

\begin{figure}[!t]
	\begin{center}
		\vspace{-0.5cm}
		\includegraphics[width=12cm]{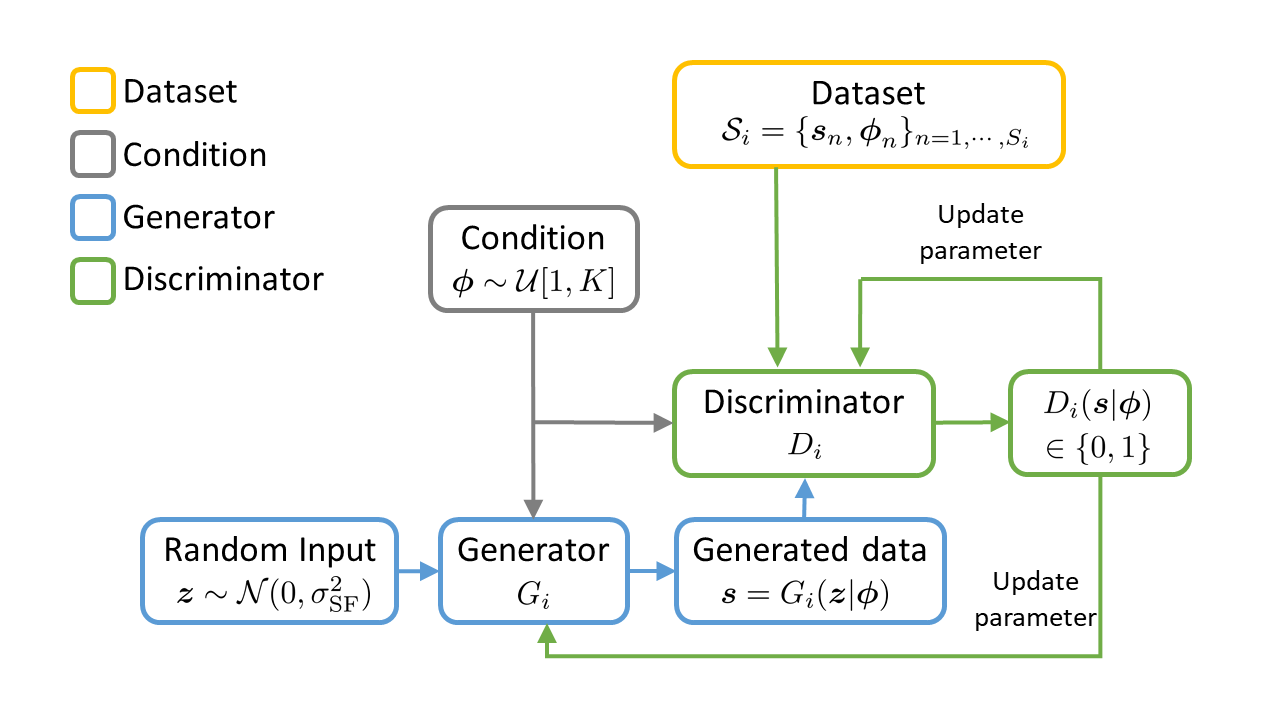}
		\vspace{-0.6cm}
		\caption{\label{CGAN} The learning framework of a stand-alone CGAN for each UAV $i$. }  
	\end{center}\vspace{-1.2cm}  
\end{figure}

Given that the pre-determined codebook defines a dedicated set of AoA-AoD pairs for each UAV and its downlink UE, the antenna steering information $\boldsymbol{\phi}$  can be considered as a prior knowledge that does not depend on channel measurement. 
Then, in order to model the channel distribution given different AoA-AoD directions, 
a conditional generative adversarial network (CGAN) framework \cite{mirza2014conditional} is applied, where each UAV $i$ has a condition sampler $\mathcal{U}$, a generator $G_i$, a discriminator $D_i$, and a local dataset $\mathcal{S}_i$, as shown in Fig. \ref{CGAN}.   
In each training epoch, the condition sampler draws an AoA-AoD pair   out of  $K$ possible directions, following a uniform distribution  $\boldsymbol{\phi} \sim \mathcal{U}[1,K]$, which is identical for each UAV, and the sampling result $\boldsymbol{\phi}$ will be used as the direction  condition in  the  CGAN training.   
Next, the generator $G_i(\boldsymbol{z}, \boldsymbol{\theta}^g_{i}|\boldsymbol{\phi})$, which is a DNN with a parameter vector $\boldsymbol{\theta}^g_i$,  maps a random input $\boldsymbol{z}$ to the channel sample space $\mathcal{S}$  under the condition $\boldsymbol{\phi}$, and the discriminator  $D_i(\boldsymbol{s},\boldsymbol{\theta}^d_i|\boldsymbol{\phi})$, which is another DNN with a parameter vector $\boldsymbol{\theta}^d_i$, takes a channel sample  $\boldsymbol{s}$ and the condition $\boldsymbol{\phi}$ as an input,  and outputs a value between $0$ and $1$. 
If the output of $D_i$ is close to $1$, then the input sample $\boldsymbol{s} = (\boldsymbol{x},\boldsymbol{y},t,\alpha)$ is highly likely to be a real data sample, which contains the channel gain $\alpha$ that is measured between the UAV located at $\boldsymbol{x}$ and the UE located at $\boldsymbol{y}$ during the time interval $t$ with an antenna steering pair $\boldsymbol{\phi}$; 
otherwise, a zero output of $D_i$ means that the input channel sample is fake. 
Therefore, the generator of each UAV $i$ aims to generate channel samples close to the real measurement data, while the discriminator tries to distinguish the fake data  from the real channel samples. 

Let $f_i$ be the channel sample distribution in the dataset $\mathcal{S}_i$, $f^G_{i}$ be the learned distribution of the generator for UAV $i$, and $f_i^z$ be the sampling distribution of the random input $\boldsymbol{z}$. 
Then, the goal of a stand-alone CGAN is to train its generator $G_i$ to find the channel distribution $f_i$ under each  condition $\boldsymbol{\phi}_k$, while the discriminator $D_i$ is used to quantify the learning accuracy of $G_i$.  
Hence, we model the interactions between the generator and discriminator of UAV $i$ by a zero-sum game framework with a value function \cite{mirza2014conditional}:
\begin{equation}\label{standaloneV}
	V_i(D_i, G_i) = \frac{1}{K} \sum_{k=1}^K \mathbb{E}_{\boldsymbol{s}\sim f_i  } \Big[  \log D_i(\boldsymbol{s}|\boldsymbol{\phi}_k) \Big] + 	\mathbb{E}_{\boldsymbol{z} \sim f^z_i} \Big[ \log(1-D_i(G_i(\boldsymbol{z}|\boldsymbol{\phi}_k))) \Big]. 
\end{equation}
For each condition $\boldsymbol{\phi}_k$, the first term in (\ref{standaloneV}) forces the discriminator to output one for the real data, and  the second term penalizes the generated data samples created by the  generator. 
Therefore, the generator of each UAV $i$ aims at minimizing the value function while its discriminator tries to maximize this value. 
It has been proven in \cite{goodfellow2014generative} that this stand-alone CGAN game admits a unique Nash equilibrium (NE) where $f^G_i = f_i$ and $D_i = 0.5$. 
At the NE, under each condition, the channel sample distribution of the generator is identical to the distribution of the dataset, and thus, the discriminator cannot distinguish between the generated samples and the real data.  
Instead, the discriminator will randomly output $0$ and $1$, which yields an average output  of $0.5$.

However, in practice,  each UAV only has a limited number of channel samples.  
Although a stand-alone CGAN can learn the channel representation of a UAV's local dataset, it can be biased and only feasible for a limited spatial-temporal domain. 
Once the UAV moves to an unvisited area or a new UE appears at a new location,   pilot measurement will again be necessary 
to acquire the propagation feature of the new environment  and update the channel model.  
However, both the data collection and the model update processes are time-demanding and energy-consuming for a UAV platform.  
Therefore,  to avoid repeated channel estimations within the same space, a UAV can learn the channel information from other UAVs that operated in this region. 
Here, we note that raw data exchange in a real-time manner among UAVs can yield heavy communication overhead and require a lot of energy and spectrum resource. 
Meanwhile, sharing the location-time information of  served UEs to an unauthorized UAV can raise privacy issues, especially when each UAV  belongs to a different network operator. 
Thus,  the  data exchange for mmWave A2G channel modeling in a distributed UAV network  must be communication-efficient and privacy-preserving.

\subsection{Distributed CGANs Framework} 
In order to address the challenge for  A2G mmWave  channel modeling, we propose a distributed CGAN framework, where a generative channel model is trained by each UAV to create channel samples from an underlying distribution of the overall dataset, without explicitly revealing the data distribution or showing real data samples.  
Note that, the use of distributed GANs is an emerging area of research in the machine learning community with only a handful of prior works \cite{ferdowsi2020brainstorming}, \cite{ferdowsi2019generative}, and \cite{hardy2019md}, none of which was used in the context of a wireless communication problem.
Meanwhile,  the existing works in \cite{ferdowsi2020brainstorming}, \cite{ferdowsi2019generative}, and \cite{hardy2019md} analyze their distributed GAN system, given that the network architecture is fixed and known, and each learning agent  is allowed to send any amount of data samples  to any number of other agents at each iteration, i.e., the data transmission capability over the communication link between any two agents is infinite,  and the link is highly reliable.  
However, these assumptions cannot be supported in a resource-limited and dynamic network.  
In order to enable the use of the distributed GAN into a realistic UAV system, a network formation approach  will be designed in the section to address practical communication constraints with limited information exchange.  

\begin{figure}[!t]
	\begin{center}
		\vspace{-0.5cm}
		\includegraphics[width=12cm]{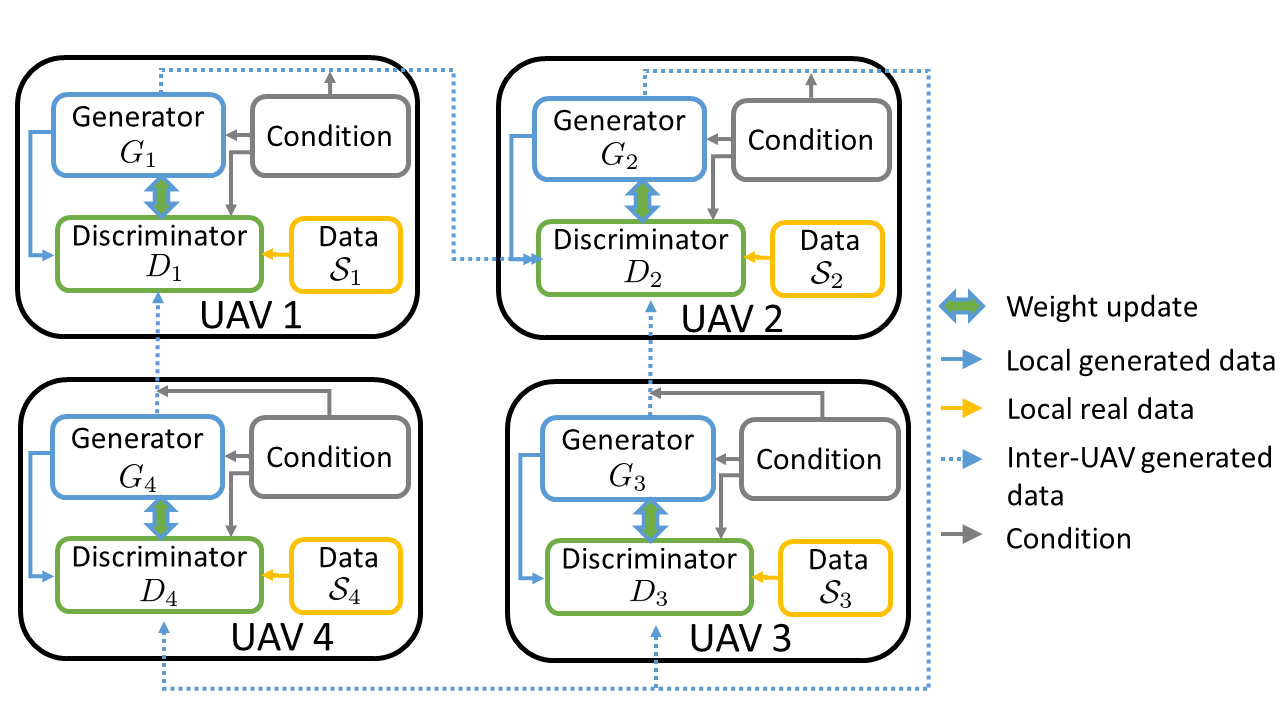}
		\vspace{-0.2cm}
		\caption{\label{learningFramework} An illustration of the distributed CGANs   framework with four UAVs, where $\mathcal{O}_1 = \{2\}$, $\mathcal{O}_2 = \{3,4\}$, $\mathcal{O}_3 = \{2\}$, and $\mathcal{O}_4 =\{1\}$. }  
	\end{center}\vspace{-1.2cm}  
\end{figure}

Given a set $\mathcal{I}$ of $I$ UAVs, we consider that the available data in $\mathcal{S}= \mathcal{S}_1 \cup \cdots \cup \mathcal{S}_I$  follow a distribution $f$. 
The local dataset $\mathcal{S}_i$ of  each UAV $i$  is collected from different geographic areas or at different service times. 
Hence, the  distribution $f_i$  of each local dataset $\mathcal{S}_i$  does not span the entire  spatial-temporal space of the real channel distribution. 
The goal is to train the generator distribution $f^G_{i}$ of each UAV $i$  to find the network-wide channel distribution $f$, under the constraint that no UAV $i$ sends its real dataset $\mathcal{S}_i$. 
To achieve this goal, we extend the newly introduced concept of \emph{distributed brainstorming GANs} in \cite{ferdowsi2020brainstorming} to the context of the studied wireless channel modeling problem with learning conditions. 
Thus, in our framework, each UAV $i$ only shares the generated samples (not the raw data) from $G_i$ and the AoA-AoD conditions with other UAVs in each training epoch.  
Fig. \ref{learningFramework} illustrated the proposed distributed CGAN framework, where the input of the discriminator for each UAV $i$ comes from its local dataset $\mathcal{S}_i$, and the generated samples from its own generator $G_i$ and the generators of  its neighboring UAVs. 
In this distributed CGAN framework, the generators collaboratively generate channel samples  to fool all of the discriminators while the discriminators try to distinguish between the generated and real channel samples.  
Let $\mathcal{N}_i$ be the set of  UAVs from whom UAV $i$ receives generated samples, and
let $\mathcal{O}_i$ be the UAV set to whom UAV $i$ sends its generated channel samples. 
Then, for each UAV $i$, the interaction between its generator and discriminator can be modeled by a game-theoretic framework with a value function:
\begin{equation} \label{indU}
\begin{aligned}
V_i(D_i, G_i, \{ G_j \}_{j \in \mathcal{N}_i}) = \frac{1}{K} \sum_{k=1}^K \mathbb{E}_{\boldsymbol{s}\sim f_i^b }  \Big[\log D_i(\boldsymbol{s}|\boldsymbol{\phi}_k)\Big] +  \mathbb{E}_{\boldsymbol{z} \sim f^z_i} \Big[\log(1-D_i(G_i(\boldsymbol{z}|\boldsymbol{\phi}_k)))\Big],
\end{aligned}
\end{equation}
where $f^b_i = \pi_i f_i +  \sum_{j \in \mathcal{N}_i} \pi_{ij}f^G_j$ is a mixture distribution of UAV $i$'s  local dataset $\mathcal{S}_i$   and received data from all neighboring UAVs in $\mathcal{N}_i$.  
Here, we define $\pi_i =  \frac{S_i}{S_i + \eta \sum_{j \in \mathcal{N}_i}S_j}$ and  $\pi_{ij}= \frac{\eta S_j}{S_i + \eta \sum_{j \in \mathcal{N}_i}S_j}$, where $ \eta {S}_j$ is the number of generated channel  samples that  UAV $j$  sends to UAV $i$ in each epoch, and $\eta >0$.  
Analogous to \cite{ferdowsi2020brainstorming}, given that the value functions of all UAVs are interdependent, we define the total utility function for the  distributed UAV  network  as follows:
\begin{equation}\label{Utility} \vspace{-0.1cm}
	V( \{D_i\}_{i=1}^{I},\{G_i\}^I_{i=1}) = \sum_{i=1}^I V_i(D_i, G_i, \{ G_j \}_{j \in \mathcal{N}_i}), 
\end{equation}
where all generators aim at minimizing the total utility function defined in (\ref{Utility}), while all discriminators try to maximize this value. 
Therefore,  the optimal discriminators and generators of the distributed CGAN learning can be derived as a  min-max problem as follows:
\begin{equation}\label{optimalDG} \vspace{-0.1cm}
	\{D^{*}_i\}_{i=1}^{I},\{G^{*}_i\}^I_{i=1} = \arg \min_{G_1,\cdots,G_I} \arg \max_{D_1,\cdots,D_I} V.  
\end{equation} 
Note that,  the optimal discriminators and generators in (\ref{optimalDG}) both depend on the structure of the UAV communication system. 
However,  the prior art in \cite{ferdowsi2020brainstorming} and \cite{ferdowsi2019generative} only defined the general  framework of distributed GANs, but it did not account for the presence of data exchange or the wireless networking optimization for  information sharing. 
Therefore, next, we will first define  the structure  of the UAV communication system. Then, based on the optimized network topology,  we identify the optimal solution to the learning problem in (\ref{optimalDG}).

\subsection{Distributed UAV Communication Network}

The communication structure of the UAV network is captured by a directed graph $\mathcal{G} = (\mathcal{I},\mathcal{E})$, where $\mathcal{I}$ is the set of  UAVs, and $\mathcal{E}$ is the set of edges. 
Each edge $e_{ij} \in \mathcal{E}$ is an ordered UAV pair that corresponds to an air-to-air (A2A) communication link. 
For example, for any $i,j\in\mathcal{I}$, if $ e_{ij}\in \mathcal{E}$, then in each CGAN learning iteration, UAV $i$ will send its generated samples to the discriminator of UAV $j$.   
The number of neighboring UAVs from whom UAV $i$ receives generated samples is called the in-degree, where  $N_i = |\mathcal{N}_i|$, and   the out-degree of UAV $i$ is $O_i = |\mathcal{O}_i|$.  
Meanwhile, for any $u,v \in\mathcal{I}$, if we can start from $u$, follow a set of connected  non-repeated edges in  $\mathcal{E}$, and finally reach $v$, then we say that a \emph{path} ${E}_{u,v}$ exists from $u$ to $v$, and the length of the path $l_{u,v}$ equals to the number of edges on ${E}_{u,v}$. 
Moreover,   we denote  the loop path that starts and ends both at $u$ as ${E}_{u}$, and denote the loop-path length as $l_u$. 

In order to efficiently share the generated channel samples, 
orthogonal frequency-division multiple access (OFDMA) techniques with $B$  available resource blocks (RBs) are used to support the A2A communication  over sub 6-GHz frequencies \cite{feng2018spectrum}, where $B \ge I$.  
In order to avoid  interference, we require  the number of communication links not to exceed the number of RBs, i.e., $|\mathcal{E}|\le B$, which is reasonable for UAV networks.      
Meanwhile, assuming a fixed hovering location for each UAV during the learning stage, the A2A communication rate from UAV $i$ to  $j$ using RB $b$ is given as   
$ R_{ij} = w_b \log_2 \left(1+ \frac{P_{ij}h_{ij}}{\sigma^2} \right)$, 
where $w_b$ is the A2A communication bandwidth, $P_{ij}$ and $h_{ij}$ are  the transmit power and path loss from UAV $i$ to  $j$,  and $\sigma^2$ is the noise power. 
A signal-to-noise ratio (SNR) threshold $\tau$ is introduced, such that for any UAV pair $(i,j)$, if the received  SNR at UAV $j$ is lower than $\tau$, i.e. $P_{ij}h_{ij}/\sigma^2 < \tau$, then, no RB will be assigned to this A2A communication link, i.e., $e_{ij} \notin \mathcal{E}$.  
In each iteration,  each UAV $i$ sends $\eta {S}_i$ generated samples to its neighboring UAVs in $\mathcal{O}_i$,   
and the transmission time for  the generated channel samples  should not exceed  $t_\tau$. 

Next, we define the \emph{completion time} $C$ of the distributed CGAN approach as  the expected number of iterations to finish the learning process, multiplied by the duration of each learning iteration \cite{ma2017distributed}.  
Here, the criterion of  learning completion is $ f^G_i = f$, $\forall i$, i.e., the generator $G_i$ of each UAV $i$ learns the entire channel distribution $f$ of the distributed network.     
To facilitate the analysis, we consider  a fixed size for each UAV’s dataset, i.e. $S_1 = \cdots = S_I = S$, and a homogeneous UAV communication network, where $N_1= \cdots = N_I = N$, so as to guarantee a synchronous learning speed.     
Meanwhile, we define  $\epsilon \in [0,1)$ to be the training error of the local discriminator at each UAV, where    
the value of $\epsilon$ represents the percentage  that a generated sample from $G_i$ does not follow the distribution of real data, but the discriminator $D_i$ does not distinguish this failure. 
Then, the probability that the learning process completes after $T$ iterations  can be derived as follows.     
\begin{theorem}\label{prob}  
	 Given the UAV network structure $\mathcal{G}$, the probability $p_{\mathcal{G}}(T)$ that the generator  $G_i$ of each UAV $i$ learns the entire channel distribution $f$ after  $T$ iterations will be given by: 
	 \begin{equation*}
	 p_{\mathcal{G}}(T)  = 
	 \begin{cases} 
	 	0  & 0< T < l^{\textrm{\normalfont max}}, \\
	 	\frac{[(1-\epsilon)\eta]^{l^{\textrm{\normalfont max}}}}{(1+N\eta)^{l^{\textrm{\normalfont max}}-1}}  & T = l^{\textrm{\normalfont max}},  \\
		p_{\mathcal{G}}(l^{\textrm{\normalfont max}})  +   \sum_{i=l^{\textrm{\normalfont max}}+1}^T \left[  \prod_{j=l^{\textrm{\normalfont max}}}^{i-1}\left( 1- \frac{ [(1-\epsilon)\eta]^{l^{\textrm{\normalfont max}}}}{(1+N\eta)^{j-1}}  \right)\right]  \frac{  [(1-\epsilon)\eta]^{l^{\textrm{\normalfont max}}}}{(1+N\eta)^{i-1}}   & l^{\textrm{\normalfont max}} < T < l^{\textrm{\normalfont max}} + l^{\textrm{\normalfont min}}_{\textrm{\normalfont loop}},
	 \end{cases}	 
	 \end{equation*}
 and for $ T \ge l^{\textrm{\normalfont max}} + l^{\textrm{\normalfont min}}_{\textrm{\normalfont loop}}$,   $	p_{\mathcal{G}}(T) =  p_{\mathcal{G}}(l^{\textrm{\normalfont max}}+l^{\textrm{\normalfont min}}_{\textrm{\normalfont loop}}-1)  + $
 \begin{equation*}
 	\begin{aligned}
 	\sum_{i=l^{\textrm{\normalfont max}}+l^{\textrm{\normalfont min}}_{\textrm{\normalfont loop}}}^T \left[  \prod_{j=l^{\textrm{\normalfont max}}+l^{\textrm{\normalfont min}}_{\textrm{\normalfont loop}}-1}^{i-1}\left( 1- \frac{  [(1-\epsilon)\eta]^{l^{\textrm{\normalfont max}}}}{(1+N\eta)^{j-1}} \prod_{k=l^{\textrm{\normalfont max}} + l^{\textrm{\normalfont min}}_{\textrm{\normalfont loop}} -1} ^j \gamma(k) \right)\right]  \frac{ [(1-\epsilon)\eta]^{l^{\textrm{\normalfont max}}}}{(1+N\eta)^{i-1}} \prod_{l=l^{\textrm{\normalfont max}} + l^{\textrm{\normalfont min}}_{\textrm{\normalfont loop}} } ^i \gamma(l), 
 	\end{aligned} 
 \end{equation*}
where $l^{\textrm{\normalfont max}} = \max_{u,v \in \mathcal{I}} l_{u,v}$ is the length of the maximum shortest-path  in $\mathcal{G}$, $l^{\textrm{\normalfont min}}_{\textrm{\normalfont loop}} $ is the length of the shortest loop-path with the same starting UAV as $l^{\textrm{\normalfont max}}$, and $\gamma(T) \ge 1$ is an acceleration coefficient. 
\end{theorem}
\begin{proof}
See Appendix \ref{appendicesA}. 	 
\end{proof}
Theorem \ref{prob} shows that the number of iterations need for learning completion is always greater than or equal to the maximum shortest-path length $l^{\textrm{max}}$ in $\mathcal{G}$. Meanwhile, the completion probability $p_{\mathcal{G}}(T) $ for each iteration $T$ will decrease as $ l^{\textrm{max}}$ becomes larger. 
Therefore,  to optimize the learning speed  in the UAV  network, it is necessary to minimize the maximum length of shortest paths among all UAVs. 
Next, based on Theorem \ref{prob}, the number of  iterations $T_\mathcal{G} \in \mathbb{N}^+$, which is required for the distributed CGAN learning to complete with a confidence level $p_\tau \in(0,1)$, is given by 
\begin{equation} \label{convergenceIteration}
p_{\mathcal{G}}(T_\mathcal{G}-1)< p_\tau \le    p_{\mathcal{G}}(T_\mathcal{G}). 
\end{equation}
Hence, after  $T_\mathcal{G}$  iterations,  the generator of each UAV is expected to learn the entire channel distribution with a probability  $p_\tau$.     
Meanwhile, we define  $t_{\epsilon}$ to be the upper-bound time for each UAV  to achieve a training error of its discriminator that is no larger than $\epsilon$. 
Then, given the network structure $\mathcal{G}$, the overall expected time to complete the distributed CGAN learning is\cite{ma2017distributed}   
\begin{equation}\label{convergenceTime}
C(\mathcal{G}) = (t_\tau  + t_\epsilon) \cdot T_\mathcal{G}. 
\end{equation}
Consequently, in the distributed UAV network with limited communication resources, the objective for the cooperative mmWave channel modeling is to form an optimal A2A communication network $\mathcal{G}$, such that the expected completion time of the distributed CGAN learning is minimized, i.e., 
\begin{subequations}\label{equsOpt}  \vspace{-0.5cm}
	\begin{align}
	\min_{\mathcal{G}} \quad &  C(\mathcal{G})  \label{equOpt}\\ 
	\textrm{s.t.} \quad  
	& \sum_{e_{ij} \in \mathcal{E}} P_{ij} \le P_\textrm{max},  \quad \forall i \in \mathcal{I},  \label{consPower}\\ 
	& P_{ij}h_{ij}/\sigma^2 \ge \tau, \quad \forall e_{ij} \in \mathcal{E}, \label{consSNR}\\
	&  \eta S_i \rho /R_{ij}  \le t_\tau,  \quad \forall e_{ij} \in \mathcal{E},  \label{consTime} \\
	& \exists E_{i,j}  \subset \mathcal{E}  , \quad \forall i,j \in \mathcal{I}, \label{consPath} \\
	& I \le|\mathcal{E}| \le B,  \label{consEdge} 	
	\end{align}
\end{subequations}
where $P_\textrm{max}$ is the maximum transmit power  for A2A communications, and $\rho$ is the data size for each channel sample.  
Here, (\ref{consPower}) limits the maximum transmit power  $P_\textrm{max}$  for each UAV,   (\ref{consSNR}) and (\ref{consTime}) set thresholds for the received SNR and the transmission time of each A2A communication link,  (\ref{consPath}) requires a strongly connected network in $\mathcal{G}$ such that each local channel dataset can be learned by all the other UAVs via the distributed CGAN framework, and (\ref{consEdge}) avoids the interference over A2A communication links.   
Here, it is worthy noting that our goal is not to find the shortest paths in the graph $\mathcal{G}$, but to identify the optimal topology $\mathcal{G}^{*}$, such that  $\mathcal{G}^{*}$ yields a minimal value, among all possible topologies, of  the maximum shortest-path between any two UAVs. 
However, in order to solve (\ref{equsOpt}),  a central controller is required to optimize the  communication structure based on path loss  between each UAV pair. 
However, in a distributed UAV network,  such a centralized entity is often not available, which makes (\ref{equsOpt}) very challenging to solve. 
To solve this problem in a distributed manner,  we must equivalently disassemble  (\ref{equsOpt}) into a set of sub-problems for each UAV, which will be detailed in the following section.
Once (\ref{equsOpt}) is solved, we subsequently derive the NE for the distributed  learning in (\ref{optimalDG}).

\section{Optimal learning for distributed CGANs } \label{solution}

In this section, we aim to jointly solve the optimization problems in (\ref{optimalDG}) and (\ref{equsOpt}) to enable an efficient  channel modeling approach using the distributed CGAN framework. 
First, we  derive  the optimal structure $\mathcal{G}^{*}$ for the UAV communication network that minimizes the expected iterations for learning completion  in (\ref{equsOpt}). 
Next, 
given the UAV network topology  $\mathcal{G}^*$,  we analytically derive the optimal distributed CGAN solution $(G_i^*,D_i^*)$ for each UAV $i$. 

\subsection{Optimal network structure for A2A UAV communications } \label{optSec1} 
In order to optimally solve (\ref{equsOpt}) in a distributed manner without a central controller, we first consider a simple scenario, where  constraint (\ref{consEdge}) is simplified as 
\begin{equation}\label{conEdgesimp}
	I = |\mathcal{E}|  \le B,
\end{equation} 
i.e., the number of A2A communication links equals to  the number of UAVs. 
Then, based on  constraints (\ref{consPath})  and  (\ref{conEdgesimp}), we derive the property of the UAV network structure as follows.  
\begin{theorem}\label{gStructure}
	Under  constraint (\ref{conEdgesimp}), 
	the strongly connected network must have a ring structure, i.e.,   $N_i = O_i = 1$, $\mathcal{N}_i \cap \mathcal{N}_j = \emptyset$, and $\mathcal{O}_i \cap \mathcal{O}_j = \emptyset$,   $\forall i,j \in \mathcal{I}$ and $i \ne j$.
\end{theorem}
\begin{proof}
	See Appendix \ref{appendicesB}.
\end{proof}
Theorem \ref{gStructure} shows that, given constraints (\ref{consPath})  and (\ref{conEdgesimp}), the network structure of the UAV communication system must be a ring, where each UAV receives the channel sample from one UAV, and sends its generated data to another UAV.  
Based on Theorems \ref{prob} and  \ref{gStructure},  we can equivalently reformulate (\ref{equsOpt}) into a set of distributed optimization problems, such that  the objective of each UAV $i$ is to choose the optimal single UAV $ \mathcal{O}_i =  \{ o_i \}$ to whom UAV $i$ sends its generated channel samples, so as to minimize the learning completion time over its maximum shortest-path while satisfying constraints (\ref{consPower})-(\ref{consTime}), i.e.,  
\begin{subequations}\label{equsOpt2}  
	\begin{align}
	\min_{o_i \in \mathcal{I}_{-i}} \quad &  l_i^{\textrm{max}} (\mathcal{G} + e_{i,o_i})  \label{equOpt2}\\ 
	\textrm{s. t.} \quad  
	& P_{i,o_i} \le P_\textrm{max},    \label{consPower2}\\ 
	& P_{i,o_i}h_{i,o_i}/\sigma^2 \ge \tau,   \label{consSNR2}\\
	&  \eta S_i \rho /R_{i,o_i} \le t_\tau,  \label{consTime2}  	
	\end{align}
\end{subequations}
where $\mathcal{I}_{-i}$ is the set of UAVs except for  $i$, $\mathcal{G} + e_{i,o_i}$ is the  graph structure generated by adding an edge $e_{i,o_i}$ to $\mathcal{G}$, and $l_i^{\textrm{max}} $ is the  maximum shortest-path  from UAV $i$ to any other UAVs. 
We define the set of feasible UAVs  to whom  UAV $i$ can send its generated channel samples  while satisfying constraints (\ref{consPower2})-(\ref{consTime2})  as 
\begin{equation}\label{feasiableSet}
\mathcal{J}_i = \{ j\in \mathcal{I}_{-i} | P_{ij} \le P_\textrm{max},   P_{ij} h_{ij}/\sigma^2 \ge \tau, \eta S_i \rho / R_{ij} \le t_\tau  \}. 
\end{equation}
Then, the necessary condition for a feasible solution to (\ref{equsOpt2}) is  provided next.  
\begin{proposition}[Necessary condition]\label{existence}
	 Under constraint (\ref{conEdgesimp}), a feasible UAV network structure $\mathcal{G}^{'}$  exists,  only if  $\bigcup_{i=1}^{I} \mathcal{J}_i = \mathcal{I}$ and $\forall i,\mathcal{J}_i \ne \emptyset$ hold.  
\end{proposition} 
\begin{proof}
See Appendix \ref{appendicesC}. 
\end{proof} 
Proposition \ref{existence} shows that if the union of feasible sets does not cover all UAVs,  the UAV network cannot form a strongly connected graph, and then, a feasible solution to   (\ref{equsOpt2}) does not exist.      
Based on Proposition \ref{existence} and Theorem \ref{gStructure}, 
we derive the sufficient condition for the optimal network structure  as follows.  
\begin{proposition}[Sufficient condition]\label{optimalResult}
	Under constraint (\ref{conEdgesimp}), given that $\bigcup_{i=1}^{I} \mathcal{J}_i = \mathcal{I}$ and $\mathcal{J}_i \ne \emptyset$ hold for all $i\in\mathcal{I}$,  the optimal UAV network structure  is  $\mathcal{G}^{'} = (\mathcal{I},\mathcal{E})$, where $\mathcal{E} \subseteq \{ e_{ij}| i\in \mathcal{I}, j\in \mathcal{J}_i  \}$ and $l_i^{\textrm{max}}(\mathcal{G}^{'}) = I-1$, $\forall i \in \mathcal{I}$.  
\end{proposition} 
\begin{proof}
See Appendix \ref{appendicesD}. 
\end{proof} 

Consequently, under constraints (\ref{consPower})-(\ref{consPath}) and (\ref{conEdgesimp}), the optimal UAV network   $\mathcal{G}^{'}$  that minimizes the expected  iterations for learning completion has a ring  structure with a  communication link set  $\mathcal{E} \subseteq \{ e_{ij}| i\in \mathcal{I}, j\in \mathcal{J}_i  \}$. 
The proof of Proposition \ref{optimalResult} in Appendix \ref{appendicesD} has shown that $\mathcal{G}^{'}$ is not only the optimal solution to the simplified problem (\ref{equsOpt2}), but also a feasible solution to our original problem (\ref{equsOpt}).  
Note that,  the solution in Proposition \ref{optimalResult}  requires that $O_i=1$, while the original constraint (\ref{consEdge}) allows the value of $O_i$ to be greater than one, i.e., the number of edges can be greater than or equal to the number of UAVs. 
Given that the total number of UAV communication edges is $|\mathcal{E}| = \sum_{i \in \mathcal{I}}O_i$,   constraint  (\ref{consEdge}) is equivalent to $\sum_{i \in \mathcal{I}}O_i \le B$.    
Next, we  derive the optimal solution for the UAV network structure given that $O_i \ge 1$. 

To find the optimal topology  $\mathcal{G}^{*}$ that solves (\ref{equsOpt}), we first start  with the feasible solution  $\mathcal{G}^{'}$, where
the strongly connected  property in (\ref{consPath})  is satisfied. Then,  more communication edges need to be added to $\mathcal{G}^{'}$ to reduce $C(\mathcal{G})$. 
Here, similarly to (\ref{feasiableSet}), we define another feasible set for each UAV $i$ that satisfies   (\ref{consPower}) - (\ref{consTime}) and (\ref{consEdge}) as 
\begin{equation}\label{feasiableSetNew}
	\hat{\mathcal{J}}_i =  \left\{ j\in \mathcal{J}_{i} \Bigg\rvert \sum_{j \in \hat{\mathcal{J}}^O_i } P_{ij} \le P_\textrm{max},   P_{ij} h_{ij}/\sigma^2 \ge \tau, \eta S_i \rho / R_{ij} \le t_\tau  \right\}, 
\end{equation}
where $\hat{\mathcal{J}}^O_i \subseteq \hat{\mathcal{J}}_i \subseteq \mathcal{J}_i$, and $\hat{\mathcal{J}}^O_i$ is a subset that contains any $O_i$ components of ${\mathcal{J}}_i$. 
Note that, when $O_i=1$, $\mathcal{J}_i = \hat{\mathcal{J}}_i = \hat{\mathcal{J}}^O_i $. 
Based on the  feasible set in (\ref{feasiableSetNew}), we derive the necessary conditions for the optimal solution to  (\ref{equsOpt})   as follows. 

\begin{corollary}[Necessary condition]\label{optFinalSolution}
	 The optimal network structure $\mathcal{G}^{*} = (\mathcal{I},\mathcal{E})$ that maximizes the probability of learning completion in (\ref{equsOpt}) requires that  $\mathcal{G}^{'} \subseteq \mathcal{G}^{*}$,  and $\mathcal{E} = \{ e_{ij}| \forall i\in \mathcal{I}, \forall j\in \hat{\mathcal{J}}^O_i  \}$.  
\end{corollary}
\begin{proof}
	See Appendix \ref{appendicesE}.
\end{proof}
Corollary \ref{optFinalSolution} shows that the optimal  network structure  $\mathcal{G}^{*}$ must include the feasible solution $\mathcal{G}^{'}$ in its topology. Meanwhile, the edge set of  $\mathcal{G}^{*}$  consists of the feasible set of each UAV $i$ with exact $O_i$ components. 
Based on Corollary \ref{optFinalSolution}, the formation approaches for the edge set $\hat{\mathcal{J}}^O_i$ and the optimal UAV network $\mathcal{G}^{*}$ is summarized in Algorithm 1. 
The proposed network formation algorithm  has a complexity of $\mathbb{O}(I^{O_i})$ for each UAV $i$, where $\mathbb{O}$ is the big-O notation.  
This complexity is reasonable because the number of available RBs $O_i$ in  distributed UAV networks is usually small. 
Meanwhile, our proposed network formation algorithm is distributed, because each UAV only needs to know its A2A channel information to other UAVs, and then, it can optimize the overall network structure by removing extra UAVs within its own feasible set. 


\subsection{Optimal learning solution for the distributed CGAN framework} \label{optSec4}

Given the optimal network structure $\mathcal{G}^{*}$, based on \cite[Proposition 1 and  Theorem 1]{ferdowsi2020brainstorming}\footnote{Although  \cite{ferdowsi2020brainstorming} does not account for the conditional learning, by treating the condition as an additional part of learning input, the analytical results in \cite{ferdowsi2020brainstorming} can be extended to support the  solution to our distributed CGAN approach \cite{mirza2014conditional}.}, the optimal generator $G_i^{*}$ of each UAV $i$ will follow the distribution 
\begin{equation}
	{f^{G}_{i}}^* = f^{b}_{i} = \pi_i f_i + \sum_{j \in \mathcal{N}_i} \pi_{ij}f^G_j, 
\end{equation}
where under each  AoA-AoD condition, the generator's distribution ${f^{G}_{i}}^*$  of  UAV $i$ equals to the mixture of the channel distribution $f_i$ from its local dataset and the generator's distribution $f^G_j$ of all incoming UAVs $j \in \mathcal{N}_i$. 
In this case, during each learning iteration, the discriminator will receive two portions of the channel samples with the same information, from its locally generated samples of $G_i^{*}$, and from the mixture source of $\mathcal{S}_i$ and $\{G_j\}_{\forall j \in \mathcal{N}_i}$. As a result, the discriminator cannot distinguish the generated samples from the real data, and, thus, it will output $1$ and $0$ with an equal   probability of $0.5$. 
Consequently, given the optimal generator $G_i^{*}$,   the output of the optimal discriminator will be 
\begin{equation}\vspace{-0.1cm}
	D_i^{*}  = \frac{f^{b}_{i}}{f^{b}_{i} + {f^{G}_{i}}^*}= \frac{1}{2}. 
\end{equation}  
Once the local adversarial training of each UAV $i$  converges  to  $(G_i^{*},D_i^{*})$, the distributed information sharing process in the UAV network  converges to a unique NE \cite{ferdowsi2020brainstorming}, and the generator of each UAV $i$ has learned the entire distribution of mmWave channels, i.e., $G_i^{*} \sim 	{f^{G}_{i}}^* =f$.  
Next, each UAV $i$ can explicitly obtain the channel distribution $f$ of the trained generative model, based on the channel samples from its optimal generator $G_i^*$.  
The data-driven approach  to identify the optimal distributed CGAN solution $(G_i^{*},D_i^{*})$ has been summarized in  Algorithm \ref{algo}.

\begin{algorithm}[t] \small   
	\caption{Network formation with distributed CGAN learning for mmWave channel modeling} \label{algo}
	\begin{algorithmic}
		\State \textbf{UAV Network Formation:} \\
		1. Each UAV $i$ measures channel information $h_{ij}$ for $j \in \mathcal{I}_{-i}$,  and broadcasts the feasible UAV set $\hat{\mathcal{J}}_i$;\\
		2. If $\bigcup_{i=1}^{I} \hat{\mathcal{J}}_i = \mathcal{I}$ and $|\hat{\mathcal{J}}_i| \ge O_i$, go to step 3; otherwise, UAVs adjust their locations, and go back to step 1; \\  
		3. Start with the network graph $\mathcal{G}$ where $\mathcal{E} =  \{ e_{ij}| i \in \mathcal{I}, j \in \hat{\mathcal{J}}_i \} $; \\
		4. \textbf{For} each UAV $i$ with $|\hat{\mathcal{J}}_i|>O_i$, \\
		\quad \quad Remove one edge $e_{ij}$ from $\mathcal{E}$ where  $ j = \min_{j \in \mathcal{J}_i} l_i^{\textrm{max}}(\mathcal{G}-e_{ij}) - l_i^{\textrm{max}}(\mathcal{G})$, while guaranteeing\\
		\quad \quad  $(\bigcup_{k\in\mathcal{I}_{-i}} \hat{\mathcal{J}}_k) \cup (\hat{\mathcal{J}}_i-j )= \mathcal{I}$,  and $\exists \hat{\mathcal{J}}^O_i \subseteq (\hat{\mathcal{J}}_i - j)$, $\sum_{k \in \hat{\mathcal{J}}^O_i}P_{ik} \le P_\textrm{max}$; \\
		\quad \textbf{Until} $|\hat{\mathcal{J}}_i|=O_i$ for all $i \in \mathcal{I}$. \\ 
		\textbf{Distributed CGAN learning:} \\
		A. Initialize $G_i$ and $D_i$ for each UAV $i \in \mathcal{I}$; \\
		B. \textbf{Repeat:} Parallel for all  $i \in \mathcal{I}$: \\ 
		\quad a. Sample  $ u$  AoA-AoD   conditions: $\boldsymbol{\phi}^{(1)},\cdots, \boldsymbol{\phi}^{(u)} \sim \mathcal{U}[1,K]$, and  $u$ random inputs: $\boldsymbol{z}^{(1)}, \cdots, \boldsymbol{z}^{(u)} \sim f^z_i$;\\
		\quad b. Generate  $ u$ channel samples $G_i(\boldsymbol{z}^{(1)}|\boldsymbol{\phi}^{(1)}), \cdots, G_i(\boldsymbol{z}^{(u)}|\boldsymbol{\phi}^{(u)}) $ from the generator of UAV $i$;\\
		\quad c. Send $\pi_{oi}u$ generated sample  to each outgoing UAV $o \in \mathcal{O}_i$, and  receive $N_i$ portions of $\pi_{ij} u$ data samples  \\ %
		\quad \quad $\{\boldsymbol{s}_j^{(1)}|\boldsymbol{\phi}^{(1)} \}, \cdots, \{\boldsymbol{s}_j^{(\pi_{ij}u)}|\boldsymbol{\phi}^{(\pi_{ij}u)} \}$  from  incoming UAVs in $\mathcal{N}_i$; \\ 
		\quad c. Sample $\pi_i u$ real channel data from local dataset: $\{\boldsymbol{s}_i^{(1)}|\boldsymbol{\phi}^{(1)} \}, \cdots,  \{\boldsymbol{s}_i^{(\pi_iu)}|\boldsymbol{\phi}^{(\pi_iu)} \}    \sim \mathcal{S}_i$;\\ 
		\quad d.  Update $\boldsymbol{\theta}_i^d$ via mini-batch stochastic gradient ascent: $\nabla_{\boldsymbol{\theta}_i^d} V(D_i(\boldsymbol{\theta}_i^d)) =$ \\ 
		\quad   $\frac{1}{2u} \nabla_{\boldsymbol{\theta}_i^d} \left[\sum_{k=1}^{\pi_i u} \log(D_i(\boldsymbol{s}_i^{(k)}|\boldsymbol{\phi}^{(k)})) +   \sum_{k=1}^{u}\log(1- D_i(G_i(\boldsymbol{z}^{(k)}|\boldsymbol{\phi}^{(k)}))) + \sum_{j \in \mathcal{N}_i} \sum_{k=1}^{\pi_{ij}u} \log(D_i(\boldsymbol{s}_j^{(k)}|\boldsymbol{\phi}^{(k)}))\right]$;\\
		\quad e. Update $\boldsymbol{\theta}_i^g$ via  mini-batch stochastic gradient descent: \\ 		
		\quad \quad 
		$\nabla_{\boldsymbol{\theta}_i^g} V(G_i(\boldsymbol{\theta}_i^g)) = \frac{1}{u} \nabla_{\boldsymbol{\theta}_i^g} \sum_{k=1}^{u} \log(1- D_i(G_i(\boldsymbol{z}^{(k)}|\boldsymbol{\phi}^{(k)}))) $; \\ 
		\quad  \textbf{Until}  convergence to the NE. 
	\end{algorithmic}
\end{algorithm}  

Here, we note that the proposed distributed CGAN approach has a communication load $\mathcal{L}  = T_{\mathcal{G}}\sum_{i\in\mathcal{I}}\eta S_i \rho O_i = T_{\mathcal{G}} \eta S \rho B$, which includes all data transmissions within the UAV network  before learning completion.  
Given that  $T_{\mathcal{G}}$ has been minimized in the optimal  topology $\mathcal{G}^*$, Algorithm \ref{algo} also guarantees a minimal communication load for the proposed  learning scheme. 
We can further adjust the communication overhead of our  distributed CGAN scheme, by adapting the value of $\eta$  to meet a large range of wireless transmission constraints.  
Meanwhile, the complexity of the local adversarial training for each UAV is similar to the original CGAN framework in \cite{mirza2014conditional}. Thus, the  complexity of our distributed learning approach is around  $T_{\mathcal{G}}$-times of the original model.  
Furthermore, compared with the FL-CGAN scheme \cite{hardy2019md}, where one central agent averages the parameters of the generator and discriminator of each UAV and send back the parameter update,  our proposed approach supports a more flexible structure that is fully distributed.  
Indeed, the use of an FL-GAN scheme requires a central controller which is not available in a UAV network. 
The training process of FL-CGAN  yields communication overhead proportional to the model size, which forbids the use of large-sized models given limited communication resources. 
However,  our proposed approach allows each UAV to employ its own neural network (NN)  architecture that can be different from other UAVs.  
Meanwhile, different from the prior art in \cite{hardy2019md} where a multi-discriminators GAN (MD-GAN) framework is developed with one centralized generator and multiple, distributed discriminators, our proposed approach enables a synchronous learning for all UAVs, thus, improving time efficiency. 
More importantly, the experimental results in \cite{ferdowsi2020brainstorming} have shown that the distributed GAN learning outperforms both the FL-GAN and MD-GAN systems, in terms of modeling accuracy and communication efficiency.

\section{Simulation Results and Analysis}\label{simulation}

For our simulations, unless state otherwise, we consider an airborne  network with $I=4$ UAVs using $B=4$ RBs to provide wireless service within a geographic area of  $400 \times 100~ \textrm{m}^2$. Each UAV has a mmWave channel dataset  that covers one of the regions  without overlap, i.e., a residential area, a city park \cite{xia2020millimeter}, an urban environment \cite{khawaja2018temporal}, and a suburban area \cite{khawaja2017uav}.   
With regards to  simulation parameters, we set $M=256$, $N=64$, $K=81$, $f = 30$ GHz, $w_b = 2$ MHz, $P_{\textrm{max}} = 40$ dBm, $\sigma^2 = -174$ dBm/Hz, $\epsilon =0.1$, $p_{\tau} = 99\%$, $\tau = 10$ dB,  $t_\tau = 0.01$ s, $\rho = 11$, $\eta = 0.5$, and $S_i = 1000$ for each UAV $i$. 
We implement a NN with  four convolution layers for the discriminator, and another NN with four transposed convolution layers   for the  generator. 

\subsection{Completion time of the distributed CGAN learning}
Fig. \ref{converge1} shows the number of required iterations to complete the distributed CGAN learning, given different numbers of available RBs  and UAVs, respectively.  
The probability of learning completion in each iteration is used as the metric to evaluate the learning rate of the proposed algorithm. 
First, as shown in Fig. \ref{convergence_RB}, given a fixed number of UAVs  $I=4$, as the number of RBs $B$ increases from $4$ to $12$, the learning speed of the proposed approach increases. 
When $B=4$,  each UAV can send its generated channel samples  in each iteration to only one neighboring UAV. 
In this case, because of the limitation on the number of  A2A communication links available,  the  efficiency of sharing channel sample is low, and the distributed learning algorithm requires a long time to complete.  
As  $B$ increases, more A2A communication links are available. When $B=12$, each UAVs can send generated channel samples to three UAVs in each iteration, which forms a fully connected system, i.e., each UAV connects with all the other UAVs for information sharing in the distributed system. 
Thus, the local data of each UAV can be shared efficiently throughout the learning framework, thus leading to a fast completion rate of $T_{\mathcal{G}}=6$ epochs\footnote{The number of iterations is not for the local CGAN training within each UAV, but for the number of times that the generated channel samples  is transmitted to the neighbors by each UAV  in the distributed airborne system.}. 

Moreover,  Fig. \ref{convergence_RB} shows that the learning process has  three distinct stages. For example, when $B=4$, the maximum shortest-path length for a four-UAVs network is $l_{\text{max}}=3$. Thus, during the learning iterations $T<3$, the probability of learning completion is always zero. Meanwhile, the minimum length of the loop path is $l_{\text{loop}}^{\text{ min}} = 4$. 
Therefore, at  iteration $T = l_{\text{max}} + l_{\text{loop}}^{\text{min}} =7$, 
the generated channel samples of each UAV begin to spread in the loop paths that connect all the UAVs in the distributed network. Thus, the learning rate after $T=7$ becomes  faster.  A similar behavior can be observed for $B=8$ when the learning rate starts to increase faster at $T = 5$, and for $B=12$, the  rate dramatically increases from $T=3$. 

Next, we show the relationship between the learning speed and the number of UAVs in   Fig. \ref{convergence_UAV}, for a fixed number of RBs. 
In a larger network,  due to a longer path length in the distributed learning system,  the number of iterations required to complete the channel modeling process also increases. 
In particular, at the beginning of the learning period, a large UAV network will experience a long and inefficient data exchange stage, which leads to a slow completion speed. 
Meanwhile, by comparing Figs. \ref{convergence_RB} and \ref{convergence_UAV}, we can see that for $I=4$ and $B=4$, the distributed learning algorithm completes learning at $T_{\mathcal{G}}=19$, while for $I=12$ and $B=12$, the required completion iterations is $T_{\mathcal{G}} > 60$. 
Although in both cases each UAV has one  RB for channel information sharing, a large UAV network size results in a much slower learning rate. 
Therefore, when the number of UAVs increases, guaranteeing an efficient channel modeling approach requires increasing the total number of RBs, as well as the average number of RBs per UAV.

\begin{figure*}[!t]
	\begin{center}
		\vspace{-1cm}
		\begin{subfigure}{.49\textwidth}
			\includegraphics[width=8.9cm]{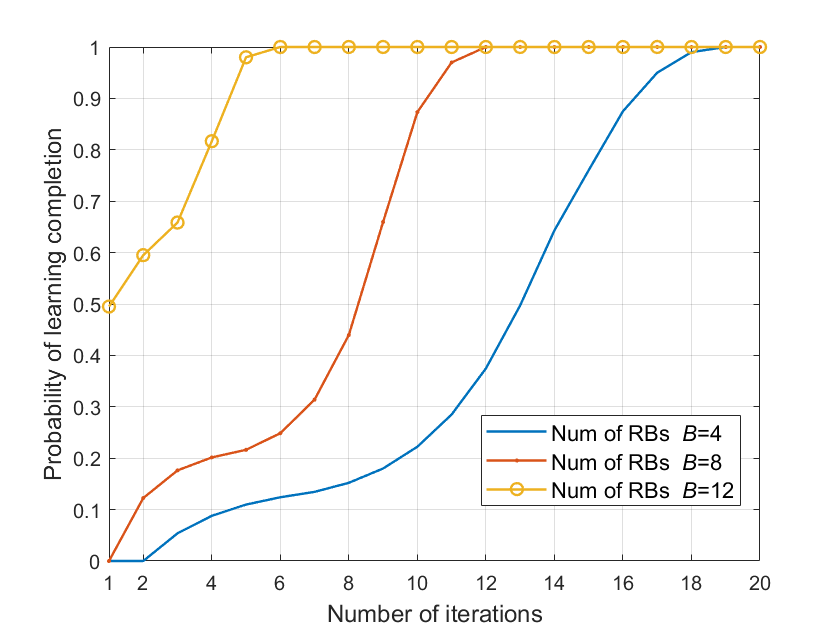}
			\caption{\label{convergence_RB}{The learning rate  increases given more RBs.}  
			}
		\end{subfigure}
		\begin{subfigure}{.49\textwidth}
			\centering
			\includegraphics[width=8.9cm]{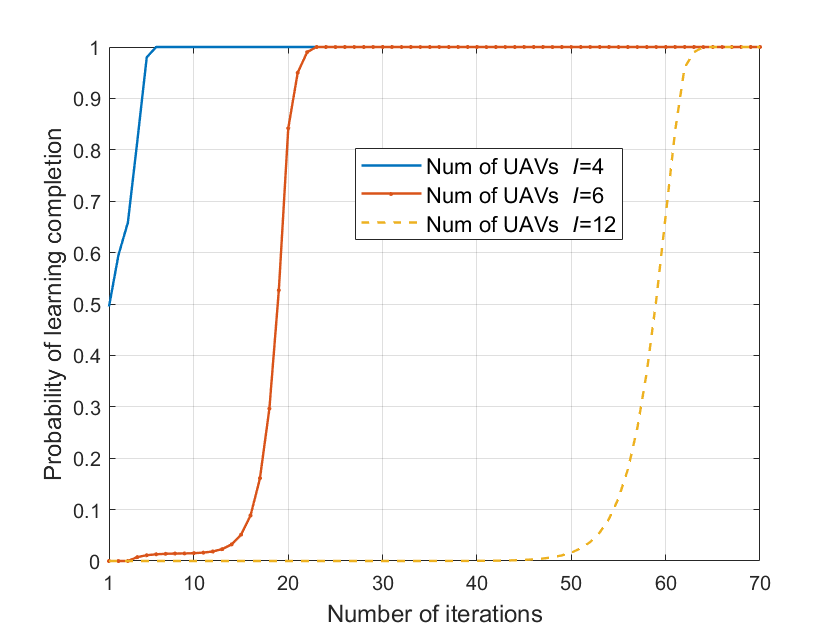}\vspace{0 cm}
			\caption{\label{convergence_UAV}{The learning rate  decreases for larger networks.}  
			}
		\end{subfigure}
		\vspace{-0.2 cm} 
		\caption{{\label{converge1} {Given more RBs for A2A communication, the proposed distributed CGAN approach yields a higher learning rate,  but this rate  decreases for a larger network with more UAVs.}  }  
		}
	\end{center}
	\vspace{- 1.2  cm}
\end{figure*}

Fig. \ref{converge2} shows  the learning rates of the distributed CGAN, for different sizes of shared data samples $\eta$ and  for different values of discriminator error $\epsilon$ at each UAV, respectively. 
Note that, in each iteration, each UAV $i$ sends $\eta S_i$  generated channel samples to its neighboring UAVs in $\mathcal{O}_i$. 
As shown in Fig. \ref{convergence_eta}, when $\eta$ becomes larger, the number of shared samples in each iteration increases, and, thus,  the learning rate of the proposed distributed CGAN approach becomes faster.  However, a larger size of the generated samples usually yields a longer transmission time. Therefore, in order to guarantee a fixed transmission duration $t_{\tau}$, a larger $\eta$ will require a better A2A channel state, or a larger A2A transmit power, so as to improve the transmission rate for A2A communications.  
Next, in Fig. \ref{convergence_epsl}, we evaluate the effect of discriminator's error $\epsilon$ in each local CGAN model to the overall learning time of the distributed learning framework.  
As the training error threshold $\epsilon$ increases from $0.01$ to $0.2$, the learning time of the distributed CGAN approach increases from $T_{\mathcal{G}} = 17$ to $20$. 
By comparing with  Figs. \ref{converge1} and \ref{convergence_eta},   the effect of the local training error $\epsilon$ on the learning rate is  very limited. 
Therefore, we can conclude that our proposed learning framework has a high tolerance to the local training error per UAV. 
Meanwhile, a high tolerance of the  training error enables each UAV to choose different NN structures, based on its own computational ability and on-board energy, for its local CGAN learning. Thus, our approach supports a very flexible distributed structure for each learning agent.    

\begin{figure*}[!t]
	\begin{center}
		\vspace{-1cm}
		\begin{subfigure}{.49\textwidth}
			\includegraphics[width=8.9cm]{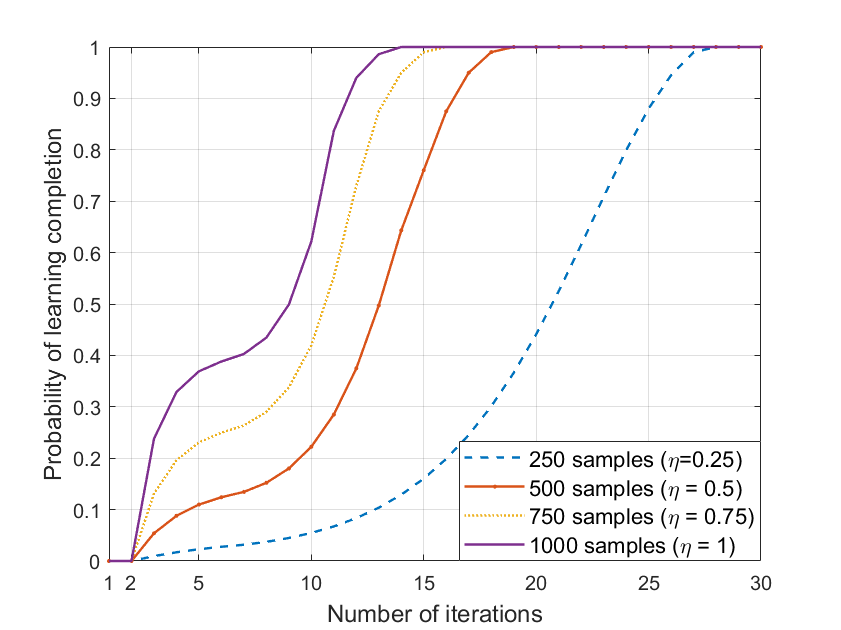}
			\caption{\label{convergence_eta} {The  learning rate increases  by sharing more generated samples in each iteration.  }  
			}
		\end{subfigure}
		\begin{subfigure}{.49\textwidth}
			\centering
			\includegraphics[width=8.9cm]{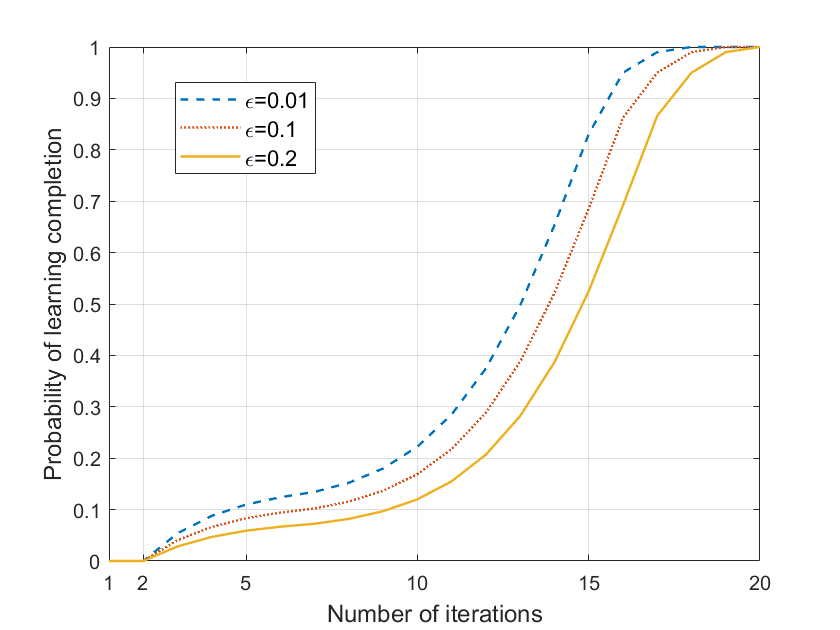}\vspace{0 cm}
			\caption{\label{convergence_epsl} {The learning rate decreases given a larger training error for each local discriminator.}  
			}
		\end{subfigure}
		\vspace{-0.2 cm} 
		\caption{{\label{converge2}{Given more generated samples and a smaller  training error, the proposed distributed CGAN approach yields a larger learning rate.}} 
		}
	\end{center}
	\vspace{- 1.2  cm}
\end{figure*} 

\subsection{Learning results for A2A mmWave channel modeling}

\begin{figure*}[!t]
	\begin{center}
		\vspace{-1cm}
		\begin{subfigure}{.49\textwidth}
			\includegraphics[width=8.9cm]{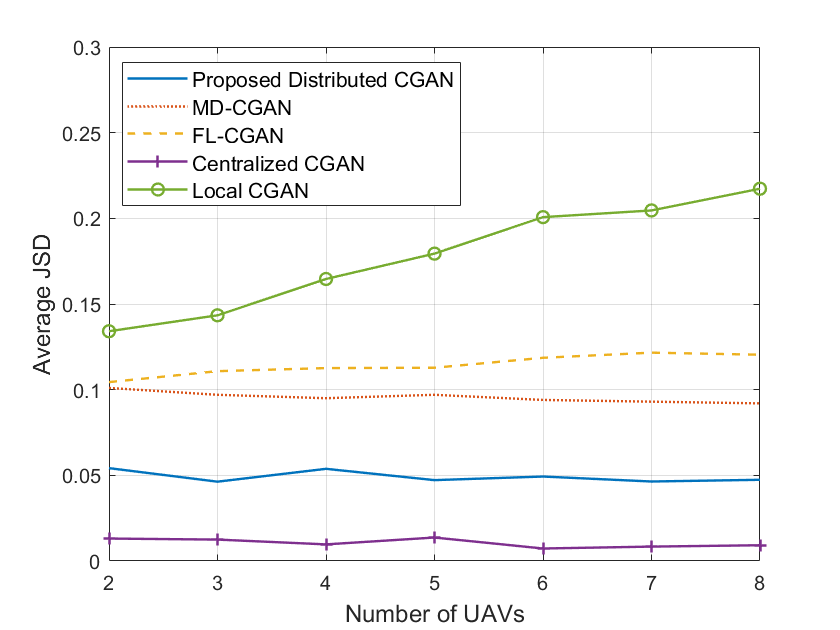}
			\caption{\label{jsd} The proposed distributed CGAN scheme outperforms other distributed and local baseline methods.  
			}
		\end{subfigure}
		\begin{subfigure}{.49\textwidth}
			\centering
			\includegraphics[width=8.9cm]{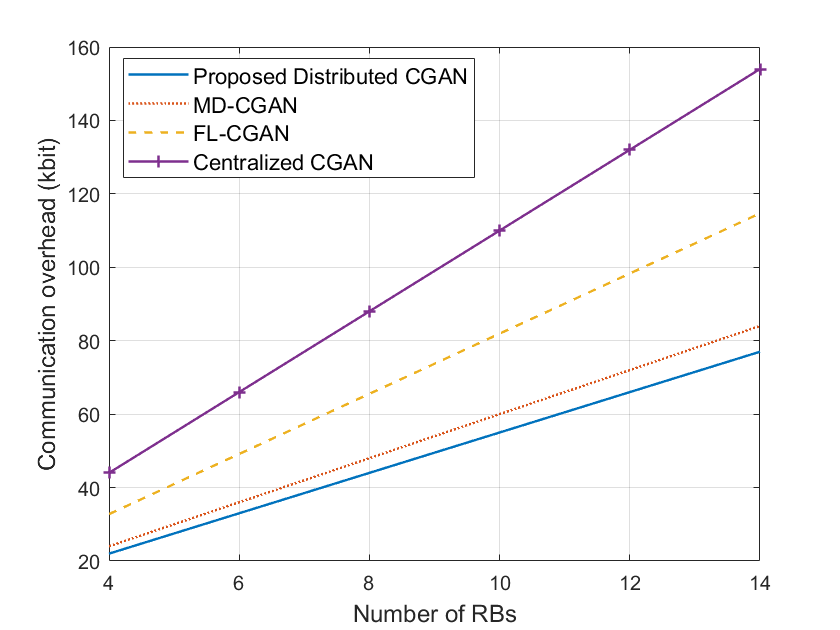}\vspace{0 cm}
			\caption{\label{commLoad} The communication overhead per iteration increases given more RBs.  
			}
		\end{subfigure}
		\vspace{-0.2 cm} 
		\caption{{\label{learning_result} Our proposed distributed CGAN approach yields the highest modeling accuracy, and the lowest communication load, compared with the distributed and local baseline schemes. }  
		}
	\end{center}
	\vspace{- 1.2  cm}
\end{figure*}  

In this section, we evaluate the learning performance of our proposed distributed CGAN approach compared with four baselines: 
a stand-alone CGAN model per UAV without information sharing, a centralized CGAN scheme based on raw channel data from all UAVs,   an FL-CGAN, and an MD-CGAN distributed learning scheme.  
After the CGAN training is completed, we calculate an empirical PDF for the learned channel distribution, by sampling a large number of channel data from the local generator  $G_i$ at each UAV $i$. 
To evaluate the modeling performance of different CGAN approaches, the criterion of Jensen-Shannon divergence (JSD) is applied to count the average distance between the real channel distribution and the learned distribution in each generator, i.e., 
$\textrm{JSD} =  \frac{1}{2I}\sum_{i \in \mathcal{I}}\sum_{s } \left[G_i(s) \log \left( \frac{G_i(s)}{F(s)} \right) +    F(s) \log \left( \frac{F(s)}{G_i(s)} \right)\right]$, 
where $I$ is the total number of UAVs, $s$ is the generated channel samples, and $F$ is the empirical PDF of the entire channel dataset. 
Here, a lower value of JSD indicates a higher learning accuracy.   
Fig. \ref{jsd}  shows that the modeling accuracy of the proposed distributed CGAN approach outperforms the local CGAN, MD-CGAN, and FL-CGAN baselines. 
First, given more UAVs in the system, each UAV covers a smaller service area, and the local generator distribution only applies to a limited spatial domain. Thus, the modeling accuracy of the local learning scheme decreases for  a larger network size.  
However, using the distributed learning approach, each UAV can learn the A2G channel property over a larger location domain from the generated samples of other UAVs. 
Thus, the modeling accuracy for all three distributed approaches stays the same for different network sizes.  
Given that our proposed distributed CGAN scheme does not require a central agent to aggregate information, it leads to a more robust structure compared with  MD-CGAN and FL-CGAN, and  it yields a higher modeling accuracy among all the distributed CGAN approaches.  
However, due to a limited training time and the inevitable training error at each UAV, the overall distributed CGAN learning of the airborne network may converge to a local optimum. 
This explains the  performance gap between the proposed distributed learning scheme and  the centralized raw data sharing method.    
To mitigate this gap, a deeper NN framework and a longer training time $t_{\epsilon}$ are needed  to decrease the training error for the local CGAN  at each UAV. 
However, those would come at the expense of additional delay or computations.

In Fig. \ref{commLoad}, we show the relationship between the communication overhead per iteration and the number of available RBs in the distributed learning network.  
Note that, the local CGAN does not involve any data sharing, thus it is not shown in the figure. 
Meanwhile, in the centralized  CGAN scheme, all of the raw channel samples for each UAV will be shared to all the other agents in the first iteration. Thus, the total number of iterations for the centralized CGAN scheme is  only one.  
Fig. \ref{commLoad} shows that 
our proposed distributed CGAN yields the lowest overhead compared with all baselines. 
For the MD-CGAN method, in each iteration, a central agent sends the generated data sample to each UAV, and then, each UAV sends the discriminator results back to the agent. 
This two-directional communication yields a  higher overhead, compared with our proposed scheme where each UAV only sends its generated samples to a dedicated set of UAVs in one direction. 
The communication overhead of the FL-CGAN scheme is proportional to the parameter size of the NN models in both the generator and discriminator. 
For channel modeling tasks, it is often the case that  the size of the CGAN parameters are larger than the size of generated samples in each iteration. Thus, the FL-CGAN method results in the highest communication load among all three distributed schemes. 

\subsection{Communication performance for online deployment }

\begin{figure}[!t]
	\begin{center}
		\vspace{-1cm}
		\includegraphics[width=8.9cm]{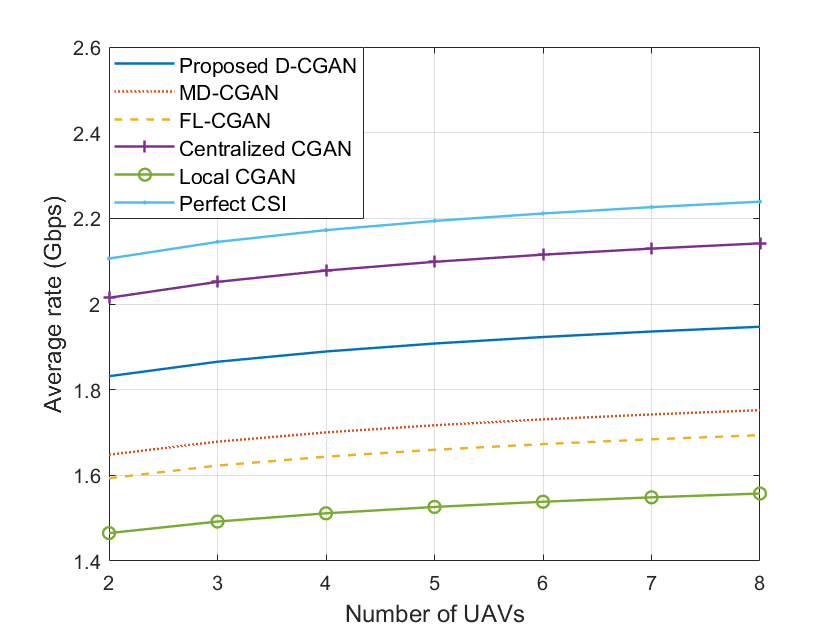}
		\vspace{-0.1cm}
		\caption{\label{rate}  {Our proposed  approach outperforms distributed and local baseline schemes, as the number of UAVs increases. }}
	\end{center}\vspace{-1.2cm}  
\end{figure}

Fig. \ref{rate} shows the online performance of the proposed channel modeling approach and three baseline schemes.    
We run each method for $1000$ times, and the average transmission rate of the UAV A2G communications is used as the performance metric, with a downlink bandwidth of $50$ MHz \cite{dilli2020analysis}.    
In order to support online deployment using the trained CGAN model, we first obtain an empirical PDF of the learned A2G channel distribution from the generator  at each UAV.   
Next, a maximum a posteriori  (MAP) estimator is used to determine the optimal beamforming and combining vectors for the UAV and its served UE, respectively, for the downlink  transmissions.   
Given more UAVs in the network, the average data rates of all schemes naturally increase, due to an averagely smaller service area for each UAV. 
Fig. \ref{rate} first shows that the proposed distributed CGAN approach improves the average data rate by around $10\%$ and $15\%$, compared with the MD-CGAN and FL-CGAN schemes, respectively.  
Meanwhile, the proposed learning scheme yields a two-fold increase in the downlink rate, compared with the local CGAN method. 
However, due to  training errors of mmWave A2G channel models,  the proposed method yields a lower data rate, compared with a perfect channel state information (CSI) scheme. 
Meanwhile,  due to a limited data transmission capability between UAVs, the proposed distributed approach  yields a lower transmit rate, compared with the centralized method. 
However, given that the perfect CSI is often not available and the UAV network does not have a central agent to build wireless links with each UAV for data collection,  the perfect-CSI scheme and centralized method are not practical solutions.

\section{Conclusion}\label{conclusion}

In this paper, we have proposed a novel framework for mmWave channel modeling  in a UAV cellular network.  
First, the channel measurement approach has been developed to collect the real-time information.  
In order to characterize mmWave A2G links in a large spatial-temporal space with different AoA-AoD directions,  a cooperative learning framework, based on the distributed CGAN, has been developed for each UAV to learn the mmWave channel distribution from other agents in a  fully-distributed manner.  
We have derived the necessary and sufficient conditions for the optimal network topology  that maximizes the learning rate, and characterized the learning solution for the generator and  discriminator per UAV.   
Simulation results have shown that the learning rate will increase by using more A2A communication RBs and sharing more generated samples in each iteration. However, a larger UAV network size and a higher training error will increase learning time. 
The results also show that the proposed CGAN approach yields a higher learning accuracy and a larger average rate for UAV downlink communications, compared with a local CGAN, MD-CGAN and FL-CGAN baseline schemes. 

\appendices
\section{Proof of Theorem 1}\label{appendicesA}
In order to facilitate our analysis and derive a closed-form expression of the  probability of the learning completion for the distributed learning framework,  we assume that the sampling process from each generator in each learning iteration is an independent process, and each generated channel sample from the same generator  contains the same amount of channel information. 
Thus, in terms of the channel information from the local dataset, each channel sampling can be considered as an i.i.d process.  
Next,  a recursion approach is used to prove Theorem 1. 
Let  $\{i, u_1, \cdots, u_{l^{\textrm{\normalfont max}}}\}$ be the ordered UAVs on the maximum-length shortest-path. 
Then, we consider a portion of channel information, whose size equals to exactly the information amount that a single channel sample can contain.
\subsection{Single-time probability}
First, we  derive the probability that this portion of information can be transmitted from the dataset $\mathcal{S}_i$ of UAV $i$ to the generator of UAV $u_{l^{\textrm{\normalfont max}}}$ just at the $T$-th iteration. 
\subsubsection{$T<l^{\textrm{\normalfont max}}$} 
In each iteration, any channel sample can only be transmitted to the next-hop UAV. Thus, when $T<l^{\textrm{\normalfont max}}$, no information can be successfully delivered from UAV $i$ to UAV $u_{l^{\textrm{\normalfont max}}}$ through the shortest path, and the probability for successful information delivery is zero.   

\subsubsection{$l^{\textrm{\normalfont max}} \le T < l^{\textrm{\normalfont max}} + l^{\textrm{\normalfont min}}_{\textrm{\normalfont loop}}$} 
During the first iteration $T=1$, UAV $i$ will send $\eta S$  generated channel samples to UAV $u_1$. 
Thus, the probability for transmitting the considered portion of information from $i$ to $u_1$ equals to the sampling ratio $\eta$. Meanwhile, given that the training error for the local generator is $\epsilon$, there is only $(1-\epsilon)$ possibility that the information can be accurately transmitted. Thus, 
the probability that the considered portion of information can be successfully transmitted from $i$ to $u_1$ $p^{\textrm{in}}_1 = (1-\epsilon)\eta$. 
At the same time, UAV $u_1$ receives generated samples from the other $N-1$ neighboring UAVs in $\mathcal{N}_{u_1}$, 
and thus, the percentage of UAV $i$'s information in the dataset of UAV $u_1$  becomes $p^{\textrm{out}}_1 =  \frac{p^{\textrm{in}}_1}{1+ N\eta} =  \frac{(1-\epsilon)\eta}{1+ N\eta}$. 
Next, when $T=2$,  UAV $u_1$ generates $\eta S$ samples and sends them to UAV $u_2$. Then, the probability that UAV $i$'s information will be transmitted from $u_1$ to $u_2$ becomes   $p^{\textrm{in}}_2 = \eta p^{\textrm{out}}_1 =  \frac{[(1-\epsilon)\eta]^2}{1+ N\eta}$. 
Due to generated data from other UAVs, the percentage of UAV $i$'s  information at UAV $u_2$ will be reduced  to $p^{\textrm{out}}_2 =    \frac{p^{\textrm{in}}_2}{1+ N\eta} = \frac{[(1-\epsilon)\eta]^2}{(1+ N\eta)^2}$. This process will continue, until the sample information is delivered to UAV $u_{l^{\textrm{max}}}$ at $T = l^{\textrm{max}}$. 
In this case, we can find the probability that UAV $i$'s information arrives at UAV $u_{l^{\textrm{max}}}$ with $l^{\textrm{max}}$ hops will be $p^{\textrm{in}}_{l^{\textrm{max}}} = \frac{[(1-\epsilon)\eta]^{l^{\textrm{max}}}}{(1+ N\eta)^{l^{\textrm{max}}-1}}$.   
In summary, the rule is that, at each iteration, the percentage of the previously owned channel samples by each UAV gets reduced by $\frac{1}{1+ N\eta}$, due to the arrival of generated samples from $N$ neighboring UAVs. Meanwhile, the  percentage of channel samples will be reduced by the sampling ratio and a training error of $(1-\epsilon)\eta$ for each hop along the path. 
Consequently, the probability that a portion of UAV $i$'s channel information can be successfully delivered to UAV $u_{l^{\textrm{max}}}$ within the $T$-th learning iteration is $p^{\textrm{in}}_{l^{\textrm{max}}}(T) = \frac{[(1-\epsilon)\eta]^{l^{\textrm{max}}}}{(1+ N\eta)^{T-1}}$. 

\subsubsection{ $T \ge l^{\textrm{\normalfont max}} + l^{\textrm{\normalfont min}}_{\textrm{\normalfont loop}}$} 
After $T\ge l^{\textrm{min}}_{\textrm{loop}}$, UAV $i$ starts receiving information of its own data distribution from its neighboring UAVs in $\mathcal{N}_i$, due to the existence of  looping data flow. 
In this case, the reduction of the information percentage in each iteration will be higher than $\frac{1}{1+N\eta}$. 
Thus, an acceleration coefficient $\gamma(T) >1$, where $\gamma(T \rightarrow +\infty) = 1+N\eta $, will be added in the iteration reduction ratio, and the probability of  information delivery becomes $p^{\textrm{in}}_{l^{\textrm{max}}}(T) = \frac{[(1-\epsilon)\eta]^{l^{\textrm{max}}} }{(1+ N\eta)^{T-1}} \prod_{i=l^{\textrm{max}} + l^{\textrm{min}}_{\textrm{loop}} } ^T \gamma(i)$ for  $T\ge l^{\textrm{min}}_{\textrm{loop}}$.

\subsection{Cumulative probability} 
Next, we will derive the accumulative probability that the considered portion of UAV $i$ has been successfully delivered to UAV $u_{l^{\textrm{max}}}$ after  $T$ iterations. 
\subsubsection{$T<l^{\textrm{\normalfont max}}$} 
Given that $p^{\textrm{in}}_{l^{\textrm{max}}}(T) = 0$, the accumulative probability $p_{\mathcal{G}}(T) = 0$ for $T<l^{\textrm{max}}$.  

\subsubsection{$T= l^{\textrm{\normalfont max}}$}
Based on the above analysis,  $p_{\mathcal{G}}( l^{\textrm{max}}) =  p^{\textrm{in}}_{l^{\textrm{max}}}(  l^{\textrm{max}}) =  \frac{[(1-\epsilon)\eta]^{l^{\textrm{max}}}}{(1+ N\eta)^{l^{\textrm{max}}-1}}$.

\subsubsection{$l^{\textrm{\normalfont max}} < T < l^{\textrm{\normalfont max}} + l^{\textrm{\normalfont min}}_{\textrm{\normalfont loop}}$} 
In this case, the accumulative probability of successful information delivery can be  calculated, based on the chain rule, where 
$p_{\mathcal{G}}(T) = p^{\textrm{in}}_{l^{\textrm{max}}}(  l^{\textrm{max}}) + [1-p^{\textrm{in}}_{l^{\textrm{max}}}(  l^{\textrm{max}})] p^{\textrm{in}}_{l^{\textrm{max}}}(  l^{\textrm{max}}+1) + \cdots + \prod_{i=l^{\textrm{max}} }^{T-1} [1-p^{\textrm{in}}_{l^{\textrm{max}}}(i)] p^{\textrm{in}}_{l^{\textrm{max}}}( T  )$. 
Consequently, we can derive  
\begin{equation*}
	p_{\mathcal{G}}(T) = p_{\mathcal{G}}(l^{\textrm{max}})  +   \sum_{i=l^{\textrm{max}}+1}^{T} \left[  \prod_{j=l^{\textrm{max}}}^{i-1}\left( 1- \frac{ [(1-\epsilon)\eta]^{l^{\textrm{max}}}}{(1+N\eta)^{j-1}}  \right)\right]  \frac{  [(1-\epsilon)\eta]^{l^{\textrm{max}}}}{(1+N\eta)^{i-1}} .
\end{equation*}

\subsubsection{$ T \ge l^{\textrm{\normalfont max}} + l^{\textrm{\normalfont min}}_{\textrm{\normalfont loop}}$} 
When the loop data flow starts, the probability becomes 
\begin{equation*}
	\begin{aligned}
		& p_{\mathcal{G}}(T) =  p_{\mathcal{G}}(l^{\textrm{max}}+l^{\textrm{min}}_{\textrm{loop}}-1)  + \\   
		&\sum_{i=l^{\textrm{max}}+l^{\textrm{min}}_{\textrm{loop}}}^T \left[  \prod_{j=l^{\textrm{max}}+l^{\textrm{min}}_{\textrm{loop}}-1}^{i-1}\left( 1- \frac{  [(1-\epsilon)\eta]^{l^{\textrm{max}}}}{(1+N\eta)^{j-1}} \prod_{k=l^{\textrm{max}} + l^{\textrm{min}}_{\textrm{loop}} -1} ^j \gamma(k) \right)\right]  \frac{ [(1-\epsilon)\eta]^{l^{\textrm{max}}}}{(1+N\eta)^{i-1}} \prod_{l=l^{\textrm{max}} + l^{\textrm{min}}_{\textrm{loop}} } ^i \gamma(l),
	\end{aligned} 
\end{equation*}
where $\gamma(l^{\textrm{max}} + l^{\textrm{min}}_{\textrm{loop}} -1) = 1$. 
This concludes our proof of Theorem \ref{prob}.

\section{Proof of Theorem 2} \label{appendicesB}

Based on the definition of a strongly connected graph, each vertex  must have at least one in-degree and one out-degree. 
Thus, according to (\ref{consPath}), a strongly connected  structure requires each UAV to at least have an incoming edge and an outgoing edge, where $N_i \ge 1$ and $O_i\ge 1$. 
Meanwhile, given  (\ref{conEdgesimp}), the number of edges equals to the number of UAVs, i.e., $\sum_{i=1}^{I} N_i = I$ and $\sum_{i=1}^{I} O_i = I$. Therefore,  we can derive that $N_i = O_i = 1$, $\forall i \in \mathcal{I}$.  
Based on the observation that each UAV has only an incoming edge and an outgoing edge, the structure of the UAV network must form a ring, so as to keep the strongly connection property, i.e., $\mathcal{N}_i \cap \mathcal{N}_j = \emptyset$, and $\mathcal{O}_i \cap \mathcal{O}_j = \emptyset$,   $\forall i,j \in \mathcal{I}$ and $i \ne j$. This concludes our proof of Theorem \ref{gStructure}. 

\section{Proof of Proposition \ref{existence}} \label{appendicesC}

We prove the necessary conditions in Proposition \ref{existence} by contradiction.   
First, we assume that there exists a UAV $i \in \mathcal{I}$, such that $\mathcal{J}_i = \emptyset$. Then, we have $O_i = 0$, which contradicts to the strong connected network requirement  where $O_i \ge 1$. 
Next, if $\sum_{i=1}^{I} \mathcal{J}_i \subset \mathcal{I}$, then, there exists at least one UAV $i$, such that none of the other UAVs sends any generated channel samples to it. Thus, for  UAV $i$, $N_i=0$, which again contradicts to the strong connected network requirement  where $N_i \ge 1$. 
Consequently,   $\bigcup_{i=1}^{I} \mathcal{J}_i = \mathcal{I}$ and $\forall i,\mathcal{J}_i \ne \emptyset$ must hold to guarantee the existence of a strongly connected structure for the optimal UAV network $\mathcal{G}^{*}$.  
 
\section{Proof of Proposition \ref{optimalResult}} \label{appendicesD}

We first prove that  Proposition \ref{optimalResult} provides a feasible solution to problems (\ref{equsOpt}) and (\ref{equsOpt2}), and then, show its optimality to problem  (\ref{equsOpt2}). 
First, given that  $\mathcal{E} \subseteq \{ e_{ij}| i\in \mathcal{I}, j\in \mathcal{J}_i  \}$, all the communication edges are formed based on the feasible UAV sets. Thus, constraints (\ref{consPower})-(\ref{consTime}), as well as (\ref{consPower2})-(\ref{consTime2}), naturally hold, according to the definition of the feasible set in (\ref{feasiableSet}). 
Meanwhile, given that $l_i^{\textrm{max}}(\mathcal{G}^{*}) = I-1$ hold for $i \in \mathcal{I}$, the UAV network must have a ring structure, which leads to a strongly connected graph that satisfies constraint (\ref{consPath}), and guarantees an equal number of communication links to the number of UAVs  that satisfies constraint (\ref{consEdge}). This concludes the proof of a feasible solution. 
Next, based on Theorem \ref{gStructure}, a ring is the only possible structure for the UAV network, where  $l^{\textrm{max}}(\mathcal{G}^{*})$ always has a fixed value of $I-1$. Given the fixed value of $l^{\textrm{max}}(\mathcal{G}^{*})$, the learning rate will be constant, for any ordered set of ring structures. Therefore, the feasible solution leads to one identical outcome of the learning time for   (\ref{equsOpt2}), and it is the optimal result, as long as the UAV network has a ring structure. 

\section{Proof of Corollary \ref{optFinalSolution}} \label{appendicesE}

First, $\mathcal{G}^{'}$ has been proved  in Proposition \ref{existence} to meet the constraint of a strongly connected network. 
Thus, the requirement $\mathcal{G}^{'} \subseteq \mathcal{G}^{*} $ will ensure that $ \mathcal{G}^{*} $ also satisfies the strongly connected property in constraint (\ref{consPath}). 
Then,  $\mathcal{E} = \{ e_{ij}| i\in \mathcal{I}, j\in \hat{\mathcal{J}}^O_i  \}$ is a necessary condition for satisfying all the other constraints in (\ref{equsOpt}), where $j\in \hat{\mathcal{J}}_i $ ensures that the communication constraints (\ref{consPower}) - (\ref{consTime}) are met according to the definition in (\ref{feasiableSetNew}), and $|\hat{\mathcal{J}}_i|=O_i$  supports a homogeneous network structure and minimizes the maximum shortest-path in $\mathcal{G}^{*}$ while satisfying the edge constraint (\ref{consEdge}).

\bibliographystyle{IEEEtran}
\bibliography{references}

\end{document}